\documentclass[11pt]{article}
\usepackage{amsfonts,amssymb,amsmath,latexsym,ae,aecompl}
\usepackage{graphics}
\usepackage{latexsym}
\usepackage{algorithm}
\usepackage{algorithmic}
\usepackage{tikz}
\usepackage{subfig}
\newcommand{\qed}{\hfill$\Box$}
\newenvironment{proof}{\noindent {\bf Proof.}}{\qed}

\newtheorem{theorem}{Theorem}[section]
\newtheorem{lemma}{Lemma}[section]
\newtheorem{corollary}{Corollary}[section]

\newtheorem{proposition}{Proposition}[section]

\begin{document}

\baselineskip 0.2in
\parskip      0.1in
\parindent    0em

\bibliographystyle{plain}

\title{{\bf Deterministic Rendezvous in Infinite Trees}}

\author{
Subhash Bhagat \footnotemark[1] \and Andrzej Pelc \footnotemark[2]
}

\date{ }
\maketitle
\def\thefootnote{\fnsymbol{footnote}}

\footnotetext[1]{
\noindent
D\'{e}partement d'informatique, Universit\'{e} du Qu\'{e}bec en Outaouais,
Gatineau, Qu\'{e}bec J8X 3X7,
 Canada.\\ {\tt subhash.bhagat.math@gmail.com}
 }

\footnotetext[2]{
\noindent
D\'{e}partement d'informatique, Universit\'{e} du Qu\'{e}bec en Outaouais,
Gatineau, Qu\'{e}bec J8X 3X7,
 Canada. {\tt pelc@uqo.ca}.
Research supported in part by NSERC  Discovery Grant 2018-03899  and by the
Research Chair in Distributed Computing of the
Universit\'{e} du Qu\'{e}bec en Outaouais.
}

\begin{abstract}
The rendezvous task calls for
two mobile agents, starting from different nodes of a network modeled as a graph to meet at the same node.
Agents have different labels which are integers from a set  $\{1,\dots,L\}$. They wake up at possibly different times and
move in synchronous rounds. In each round, an agent can either stay idle or move to an adjacent node.
We consider deterministic rendezvous algorithms. The time of such an algorithm is the number of rounds since the wakeup of the earlier agent till the meeting. 
In most of the literature concerning rendezvous in graphs, the graph is finite and the time of rendezvous depends on its size. 
This approach is impractical for very large graphs and impossible for infinite graphs. For such graphs it is natural to design rendezvous algorithms whose time depends
on the initial distance $D$ between the agents. In this paper we adopt this approach and consider rendezvous in infinite trees. 
All our algorithms work in finite trees as well.
Our main goal is to study the impact of orientation of a tree on the time of rendezvous.

We first design a rendezvous algorithm working for unoriented regular trees, whose time is in $O(z(D) \log L)$, where $z(D)$ is the size of the ball of radius $D$, i.e, the number of nodes at distance at most $D$ from a given node. The algorithm works for arbitrary delay between waking times of agents and does not require any initial information about parameters $L$ or $D$. Its disadvantage is its complexity: $z(D)$ is exponential in $D$ for any degree $d>2$ of the tree. We prove that this high complexity is inevitable: $\Omega(z(D))$ turns out to be a lower bound on rendezvous time in unoriented regular trees, even for simultaneous start and even when agents know $L$ and $D$. Then we turn attention to oriented trees. While for arbitrary delay between waking times of agents the lower bound $\Omega(z(D))$ still holds, for simultaneous start the time of rendezvous can be dramatically shortened. We show that if agents know either a polynomial upper bound on $L$ or a linear upper bound on $D$, then rendezvous can be accomplished in oriented trees in time $O(D\log L)$, which is optimal. If no such extra knowledge is available, a significant speedup is still possible: in this case we design an algorithm working in time $O(D^2+\log ^2L)$.

\vspace*{1cm}


\end{abstract}

\pagebreak

\section{Introduction}

The rendezvous task calls for
two mobile agents, starting from different nodes of a network modeled as a graph to meet at the same node.
This task is ubiquitous in many applications. In the social context, people may want to meet in a city whose streets form a graph.
In computer networks, software agents navigating in a network have to meet to share data collected from distributed databases.
In robotics, mobile robots  circulating in a network of corridors in a mine or a building may have to meet to coordinate maintenance tasks. 

In most of the literature concerning rendezvous in graphs, the graph is finite and the time of rendezvous depends on its size. 
This approach is impractical for very large graphs and impossible for infinite graphs. For such graphs it is natural to design rendezvous algorithms whose time depends
on the initial distance $D$ between the agents. In this paper we adopt this approach and consider rendezvous in infinite trees. 
All our algorithms work in finite trees as well.
Our main goal is to study the impact of orientation of a tree on the time of rendezvous.

\subsection{The model}

We consider infinite trees. %Note that each agent knows $d$ from the outset, as it sees the degree of its starting node.
They can be either {\em unoriented} or {\em oriented}. An unoriented tree does not have labels of nodes and the port labelings at each node are arbitrary. In oriented trees, one node, called the {\em root}, has label $R$, all other nodes do not have labels, and at each node different from $R$ the port 0 is on the simple path toward $R$.

Agents have different labels which are integers from a set  $\{1,\dots,L\}$. They wake up at possibly different times and
move in synchronous rounds. In each round, an agent can either stay idle or move to an adjacent node. An agent makes a move by choosing a port number at its current node. When entering the adjacent node corresponding to the chosen edge the agent learns the port of entry and the degree of this node.
We assume that the memory of the agents is
unlimited: from the computational point of view they are modeled as Turing machines.
We consider deterministic rendezvous algorithms. Both agents execute the same algorithm, each agent starting in its wakeup round. Each agent knows its label which is a parameter of the algorithm but does not know the label of the other agent.
The execution time of such an algorithm is the number of rounds since the wakeup of the earlier agent till the meeting. This time depends on values of the initial distance $D$ between the agents and  of the size $L$ of the space of labels. These values may be known or unknown to the agents, depending on the scenario. By saying that the time of an algorithm is $f(D,L)$ we mean that this is the worst case time of the execution of this algorithm, over all pairs of starting nodes of agents at distance $D$, over all pairs of agents' labels from $\{1,\dots,L\}$, and over all possible delays between waking times of agents, if the algorithm works for arbitrary delay.

\subsection{Our results}

We first design a rendezvous algorithm working for unoriented regular\footnote{The discussion of the assumption of regularity is deferred to the Conclusion.} trees of degree $d\geq 2$.  These are infinite trees all of whose nodes have degree $d$. (For $d=2$ this is the infinite line). The time of our algorithm 
 is in $O(z(D) \log L)$, where $z(D)$ is the size of the ball of radius $D$, i.e, the number of nodes at distance at most $D$ from a given node. The algorithm works for arbitrary delay between waking times of agents and does not require any initial information about parameters $L$ or $D$. Its disadvantage is its complexity: $z(D)$ is exponential in $D$ for any degree $d>2$. We prove that this high complexity is inevitable: $\Omega(z(D))$ turns out to be a lower bound on rendezvous time in unoriented regular trees, even for simultaneous start and even when agents know $L$ and $D$. Then we turn attention to oriented trees. While for arbitrary delay between waking times of agents the lower bound $\Omega(z(D))$ still holds, for simultaneous start the time of rendezvous can be dramatically shortened. 
 Our algorithms for oriented trees do not assume regularity of the tree. 
 We show that if agents know either a polynomial upper bound on $L$ or a linear upper bound on $D$, then rendezvous can be accomplished in oriented trees in time $O(D\log L)$, which is optimal in view of \cite{DFKP}. If no such extra knowledge is available, a significant speedup is still possible: in this case we design an algorithm working in time $O(D^2+\log ^2L)$.

%--------------------------------------------------
\subsection{Related Work}
\label{subsec:relatwork}
%--------------------------------------------------

The task of rendezvous has been studied in the literature both in the randomized and deterministic settings.
 Randomized rendezvous is surveyed in
\cite{alpern02b}, cf. also  \cite{alpern95a,baston98}. 
Deterministic rendezvous in networks is surveyed in \cite{Pe,Pe2}.
Several authors
considered geometric settings (rendezvous in an interval of the real line, e.g.,  \cite{baston98,baston01},
or in the plane, e.g., \cite{anderson98b,BDPP,CGKK}).
Rendezvous of more than two agents, also called {\em gathering}, was studied, e.g., in \cite{DP,fpsw,lim96}.

In the deterministic setting, feasibility and time complexity of synchronous rendezvous in networks is one of the main topics of investigation. Deterministic rendezvous of agents equipped with tokens used to mark nodes was considered, e.g., in~\cite{KKSS}. In most of the papers concerning rendezvous in networks, nodes of the network are assumed to be unlabeled and marking nodes by agents is not allowed. 
In this case, anonymous agents cannot meet in many highly symmetric networks, e.g., in oriented rings.
Hence symmetry is usually broken by assigning the agents distinct labels and assuming that each agent knows its own label but not the label of the other agent.
Deterministic rendezvous of labeled agents in rings was investigated, e.g., in \cite{DFKP,KM} and in arbitrary graphs in  \cite{DFKP,KM,TSZ07}.
%In \cite{DFKP}, the authors gave tight upper and lower bounds
%of $\Theta (D\log \ell)$ on the 
%time of rendezvous when agents start simultaneously, where $D$ is the initial distance between agents and $\ell$ is the smaller label. 
%They also gave a lower bound of $\Omega(n+D\log \ell)$ on the time of rendezvous with arbitrary delay in $n$-node rings.  In \cite{KM} an upper bound $O(n\log \ell)$
%on the time of rendezvous was given, even without knowledge of $n$.
%In \cite{DFKP}, the authors present a rendezvous algorithm whose running time is polynomial in the size of the graph, in the length of the shorter
%label and in the delay between the starting times of the agents. In \cite{KM,TSZ07}, rendezvous time is polynomial in the first two of these parameters and independent of the delay.
Gathering many anonymous agents in unlabeled networks was studied in \cite{DP}. In this weak scenario, not all initial configurations of agents are possible to gather, and the authors of \cite{DP} characterized all such configurations and provided universal gathering algorithms for them. 
%On the other hand, time of rendezvous in labeled networks was studied, e.g., in \cite{MP}, in the context of algorithms with advice. 
In \cite{CCGKM}, the authors studied rendezvous under a very strong assumption that each agent has a map of the network and knows its position in it.
Using this assumption  they designed optimal algorithms for several classes of networks, including the infinite line and finite trees.

Another measure of efficiency of rendezvous algorithms is the amount of memory needed to execute this task.
Memory of the agents required to achieve deterministic rendezvous was studied in \cite{FP2} for trees and in  \cite{CKP} for arbitrary graphs.
Memory needed for randomized rendezvous in the ring was investigated, e.g., in~\cite{KKPM08}. 

A scenario significantly differing from the above was discussed by several authors. The difference is in  dropping the assumption that agents navigate in synchronous rounds.
Asynchronous rendezvous and approach in the plane was studied in \cite{BBDDP,CFPS,fpsw} and asynchronous rendezvous in networks modeled as graphs was investigated in
\cite{BCGIL,DGKKP,DPV}.
In the latter scenario, the agent chooses the edge to traverse, but the adversary controls the speed of the agent. Under this assumption, rendezvous
at a node cannot be guaranteed even in the two-node graph. Hence the rendezvous requirement is relaxed to permit the agents to meet inside an edge.
In \cite{BCGIL}, the authors designed almost optimal algorithms for asynchronous rendezvous in infinite multidimensional grids, under a strong assumption that the agent knows its position in the grid.
In \cite{DPV}, the authors designed a polynomial-cost algorithm for asynchronous rendezvous in arbitrary finite graphs, without this assumption.

\section{Unoriented regular trees}

In this section we design and analyze a rendezvous algorithm for agents operating in an unoriented infinite regular tree of degree $d\geq 2$.  This is an infinite tree all of whose nodes have degree $d$. (For $d=2$ this is the infinite line). Nodes do not have labels and ports at each node are arbitrarily numbered by  integers $0,1,\dots,d-1$. Note that each agent knows $d$ from the outset, as it sees the degree of its starting node.
In any such tree, we define the ball $B(v,r)$ of radius $r$, centered at node $v$, as the subtree induced by all nodes at distance at most $r$ from $v$. Let $z(r)$ be the number of nodes in any ball 
$B(v,r)$. Let $a(r)=2(z(r)-1)$. Note that $a(r)$ is the number of edge traversals of a DFS exploration of any ball $B(v,r)$, starting and finishing at the same node. 

 Labels of agents are integers from the set $\{1,\dots, L\}$. For any label $\ell \in \{1,\dots, L\}$ we define the transformation $Trans(\ell)$ of $\ell$ as follows. Lat $\sigma$ be the binary representation of $\ell$. $Trans(\ell)$  is the sequence obtained from $\sigma$ by replacing every bit 1 by the string $(010101)$ and replacing every bit 0 by the string $(101010)$. Hence  $Trans(\ell)$ is of length $O(\log L)$ and has the property that it does not contain a substring of three consecutive zeroes.
\subsection{The algorithm}

For any node $v$ and any positive integer $r$, we define the following procedures. 

{\bf Procedure $ACT(v,r)$}

\hspace*{1cm}Explore the ball $B(v,r)$ by DFS, in increasing order of port numbers at each node,\\ 
\hspace*{1cm}starting and finishing at node $v$.

{\bf Procedure $PASS(v,t)$}

\hspace*{1cm}Stay at node $v$ for $t$ rounds.

The following procedure takes as parameters the degree $d$ of the infinite regular tree, a bit $b$ of the transformed label, and a positive integer $i$.
For $d=2$, i.e., for the infinite line, and for any integer $i\geq 0$, the procedure consists of  two DFS explorations of the ball $B(v,2^{2i})$, if the bit $b$ is 1, and of staying at node $v$ for the duration of these two explorations, if the bit $b$ is 0. (Exploration of a ball of radius 0 takes time 0.) For $d>2$, the procedure consists of  two DFS explorations of the ball $B(v,2i)$, if the bit $b$ is 1, and of staying at node $v$ for the duration of these two explorations, if the bit $b$ is 0.

{\bf Procedure $Exec(d,b,i)$}

\hspace*{1cm}{\bf if} $d=2$ {\bf then} $r_i=2^{2i}$\\
\hspace*{1cm}{\bf else} $r_i=2i$\\
\hspace*{1cm}{\bf if} $b=1$ {\bf then}\\
\hspace*{2cm}$ACT(v,r_i)$\\
\hspace*{2cm}$ACT(v,r_i)$\\
\hspace*{1cm}{\bf else}\\
\hspace*{2cm}$PASS(v,2a(r_i))$\\

We first present the high-level idea of the rendezvous algorithm and its challenges.
The algorithm exploits the fact that agents have different labels and guarantees that at some point one of the agents stays idle at its starting node while the other one fully explores a sufficiently large ball centered at its own starting node and thus meets the waiting agent. Exploration and waiting periods depend on bits of the transformed label of the agent.  Since agents do not know the distance between them, the algorithm is divided into stages with increasing radii of explored balls and increasing waiting times. The main challenge is due to the fact that agents may start with some delay and that, due to possibly different label lengths, they complete the same stage at different speeds. Hence the whole process may become significantly desynchronized and the difficulty is to hold one agent idle at its starting node for a sufficiently long time to allow the other agent to meet it.

Now we are ready to succinctly describe Algorithm URT (for unoriented regular trees). 
The algorithm is executed by an agent with label $\ell$, starting at a node $v$ of an infinite regular tree of degree $d$.
The algorithm works in stages  $i=0,1,2,\dots$. In a stage $i$ it ``executes'' consecutive bits $b_j$ of $Trans(\ell)$ by performing procedure $Exec(d,b_j,i)$. Notice that for $d=2$, $Exec(d,b_j,i)$ explores balls 
$B(v,2^{2i})$ if $b_j=1$ and instructs the agent to wait a corresponding time if $b_j=0$, while for $d>2$, $Exec(d,b_j,i)$ explores balls 
$B(v,2i)$ if $b_j=1$ and instructs the agent to wait a corresponding time if $b_j=0$. This is because for $d=2$ the size of a ball is linear in its radius, while for $d>2$ it is exponential in it. Stages are organized so that their durations telescope. For technical reasons we want the size of the balls treated in stage $i+1$ to be at least 4 times larger  (and not only 2 times larger) than those in stage $i$.
The algorithm is interrupted as soon as the agents meet. 

{\bf Algorithm URT}

\hspace*{1cm}$Trans(\ell):= (b_1,b_2,\dots,b_k)$\\
\hspace*{1cm}{\bf for} $i=0,1,2,\dots$ {\bf do}\\
\hspace*{2cm}{\bf for} $j=1$ {\bf to} $k$ {\bf do}\\
\hspace*{3cm}$Exec(d,b_j,i)$

\subsection{Correctness and complexity}

In this section we prove the correctness of Algorithm URT and analyze its time complexity.
We start with the definition of the critical stage. Intuitively, it is the first stage such that the radius $r_i$ of the balls explored in this stage is at least the initial distance $D$ between the agents. Thus, more formally, the critical stage is the smallest integer $i$ such that:
\begin{itemize}
\item
$D\leq 2^{2i}$, if $d=2$
\item
$D\leq 2i$, if $d>2$.
\end{itemize}
This smallest integer $i$ is denoted by $i^*$.  Hence $r_{i^*}$ is the radius of balls explored in the critical stage.

  According to the algorithm, during an execution of bit $1$ in stage $i$, an agent explores a ball of radius $r_i$ twice. We refer to each of these explorations as one {\it activity cycle}. Thus, during an execution of a single bit $1$, an agent performs two {\it activity cycles}. Similarly, during an execution of a  bit $0$ in stage $i$, an agent waits for two consecutive periods each of duration of one such exploration.
We refer to each of these waiting periods as one {\it passivity cycle}. 
 Thus, during an execution of a single bit $0$, an agent waits for two {\it passivity cycles}.
 
 For a given $d$, let $\pi_i=2a(r_i)$ be the duration of the execution of a bit in stage $i$, using procedure $Exec(d,b,i)$.
Let $y$ be the length of the transformed label of an agent. Let $S_i=y\pi_i$ denote the duration of stage $i$. The time $\alpha$ taken by the agent to reach its critical stage is the sum of durations of all previous stages. Thus this time is
$\sum_{i=0}^{i^*-1} S_i$.

In the following lemma we compute the duration $S_i$ of stage $i$, depending on $d$ and $y$.

\begin{lemma}
\label{lem-01}
For $i\ge0$,
\[ S_{i}=
  \begin{cases}
    4^i\cdot 8y       & \quad \text{if } d=2\\
    4yd\frac{(d-1)^{2i}-1}{d-2}  & \quad \text{if } d\ge3
  \end{cases}
\]
\end{lemma}
\begin{proof}
First consider the case when $d=2$. In stage $i$, an agent explores a ball of radius $r_i=2^{2i}$. We first compute $z(r_i)$, the number of nodes in the ball of radius $r_i$. Since in this case the graph is a line, $z(r_i)=2\cdot2^{2i}+1$. Recall that, $a(r_i)=2(z(r_i)-1)$ is the number of edge traversals of a DFS exploration of the ball $B(v,r_i)$, starting and finishing at the same node. Thus, $a(r_i)=2\cdot2\cdot2^{2i}$. By definition,  $\pi_i=2a(r_i)$ i.e., $\pi_i=8\cdot 4^i$. Hence, $S_i=8y4^i$. 

Next suppose $d\ge3$. In this case, in stage $i$, an agent explores a ball of radius $r_i=2i$. The value of $z(r_i)$ is as follows:
\begin{equation*}
\begin{split}
z(r_i) & =1+d+d(d-1)+d(d-1)^2+d(d-1)^3+\cdots+d(d-1)^{2i-1}  \\
 & =1+d\{1+(d-1)+(d-1)^2+\cdots+(d-1)^{2i-2}\}\\
 & = 1+d\frac{(d-1)^{2i}-1}{d-2}
\end{split}
\end{equation*} 
Thus, in this case, $a(r_i)=2d\frac{(d-1)^{2i}-1}{d-2}$ and $\pi_i=4d\frac{(d-1)^{2i}-1}{d-2}$. Hence, $S_i=y4d\frac{(d-1)^{2i}-1}{d-2}$. 
\end{proof}

\begin{lemma}
\label{lem-02}
For $i,q\ge0$, we have
\[ S_{i+q}=
  \begin{cases}
    4^qS_i       & \quad \text{if } d=2\\
    (d-1)^{2q}S_i+S_q  & \quad \text{if } d\ge3
  \end{cases}
\]
\end{lemma}
\begin{proof}
When $d=2$, by  Lemma $\ref{lem-01}$, we have
\begin{equation*}
\begin{split}
S_{i+q}&=4^{i+q}yd\\
 & =4^q4^iyd\\
 &=4^qS_i
 \end{split}
\end{equation*} 

Next suppose $d\ge3$. In this case, by  Lemma $\ref{lem-01}$, we have,
\begin{equation*}
\begin{split}
S_{i+q}&=4yd\frac{(d-1)^{2i+2q}-1}{d-2} \\
 & =4yd\frac{(d-1)^{2i+2q}}{d-2}- \frac{4yd}{d-2}\\
 &=4yd(d-1)^{2q}\frac{(d-1)^{2i}}{d-2}- 4yd\frac{(d-1)^{2q}}{d-2}+ 4yd\frac{(d-1)^{2q}}{d-2}- \frac{4yd}{d-2}\\
 & = 4yd(d-1)^{2q}\frac{(d-1)^{2i}-1}{d-2}+ 4yd\frac{(d-1)^{2q}-1}{d-2}\\
 & =(d-1)^{2q}S_i+S_q
 \end{split}
\end{equation*} 
\end{proof}

The next lemma estimates the duration of any stage in terms of the duration of the preceding and of the following stage.

\begin{lemma}
\label{lem-03}
For $i,p\ge0$, we have
$4S_{i+p-1}\le S_{i+p}\le 5 S_{i+p-1}$
\end{lemma}

\begin{proof}
We consider two cases:
\begin{itemize}
\item $\boldsymbol{d=2}$ In this case, by Lemma $\ref{lem-02}$, $S_{i+p}=4S_{i+p-1}<5S_{i+p-1}$.
\item $\boldsymbol{d\ge3}$ By Lemma $\ref{lem-02}$, we have $S_{i+p}=(d-1)^2S_{i+p-1}+S_1$. This implies $S_{i+p}>(d-1)^2S_{i+p-1}$. Since $d\ge3$, we can conclude that $4S_{i+p-1}\le S_{i+p}$. Finally consider the following:
\begin{alignat}{2}
  &\quad
  &S_{i+p-1}
  &=(d-1)^{2i+2p-4}S_1+S_{i+p-2}  \notag\\ 
  &
  & (d-1)^2 S_{i+p-1}
  &\ge(d-1)^{2i+2p-2}S_1
  \end{alignat}
  Again, we have
  \begin{alignat}{2}
  &\quad
  &S_{i+p}
  &=(d-1)^{2i+2p-2}S_1+S_{i+p-1}  \notag\\ 
  &
  & S_{i+p}
  &\le (d-1)^2 S_{i+p-1}+S_{i+p-1} \qquad(\text{by (1)})\notag\\
  &
  & S_{i+p}
  &\le ((d-1)^2+1) S_{i+p-1}
  \end{alignat}
Since $d\ge3$, (2) implies $S_{i+p}\le 5 S_{i+p-1}$. This concludes the proof.
\end{itemize}
\end{proof}

Lemma $\ref{lem-03}$ implies

\begin{corollary}
\label{cor-01}
$4\pi_{i+p-1}\le \pi_{i+p}\le 5\pi_{i+p-1}$
\end{corollary}

%\begin{corollary}
%\label{cor-01}
%$\pi_{i+1}\ge 4\pi_i$
%\end{corollary}
%\begin{proof}
%Suppose $d=2$. Then by Lemma $\ref{lem-01}$, $S_{i+1}=4S_i$ and hence $\pi_{i+1}=4\pi_i$. On the other hand, when $d\ge3$, we have $S_{i+1}=(d-1)^2S_i+S_1$. This implies that $S_{i+1}>(d-1)^2S_i$. Since $d\ge3$, we can conclude that $\pi_{i+1}>4\pi_i$.
%\end{proof}

In the next lemma we compute the value $\alpha$ of the time taken by an agent to reach its critical stage.

\begin{lemma}
\label{lem-04}
$\alpha=\frac{8yr_{i^*}}{3}-\frac{8y}{3}$ if $d=2$, and $\alpha=\frac{S_{i^*}}{(d-1)^2-1}$ if $d\geq 3$.
%\[ \alpha
%  \begin{cases}
%    =     \frac{8yr_{i^*}}{3}-\frac{8y}{3}  & \quad \text{if } d=2\\
%    \le \frac{S_{i^*}}{(d-1)^2-1}  & \quad \text{if } d\ge3
%  \end{cases}
%\]
\end{lemma}

\begin{proof}
When $d=2$,  we have
\begin{equation*}
\begin{split}
\alpha & =S_0+S_1+S_2+\cdots+S_{i^*-1} \\
      &= 8y+4\cdot8y+4^2\cdot8y+\cdots+4^{i^*-1}8y \qquad \text{(by Lemma \ref{lem-01} for $d=2$)}   \\
 & = 8y(1+4+4^2+\ldots+4^{i^*-1})\\
 & = 8y\frac{4^{i^*}-1}{3}\\
 & =\frac{8yr_{i^*}}{3}-\frac{8y}{3}  \qquad (\because r_{i^*}=2^{2i})
\end{split}
\end{equation*}

Suppose $d\ge3$. Then we have
\begin{equation*}
\begin{split}
\alpha &= S_0+S_1+S_2+\cdots+S_{i^*-1}\\
   &= 0+4yd\frac{(d-1)^2-1}{d-2}+4yd\frac{(d-1)^4-1}{d-2}+\cdots+4yd\frac{(d-1)^{2i^*-2}-1}{d-2}\\
   &=4y\frac{d}{d-2}\{(d-1)^2-1+(d-1)^4-1+\cdots+(d-1)^{2i^*-2}-1\}\\
   &=4y\frac{d}{d-2}\{(d-1)^2+(d-1)^4+\cdots+(d-1)^{2i-2}-(i^*-1)\}\\
   &\le 4y\frac{d}{d-2}\{(d-1)^2+(d-1)^4+\cdots+(d-1)^{2i^*-2}\} \qquad (\because i^*\ge1)\\
   &\le 4y(d-1)^2\frac{d}{d-2}\{1+(d-1)^2+\cdots+(d-1)^{2i^*-4}\}\\
   &\le 4y(d-1)^2\frac{d}{d-2} \left[\frac{(d-1)^{2i^*-2}-1}{(d-1)^2-1}\right]\\
   &\le 4y\frac{d}{d-2}\left[\frac{(d-1)^{2i^*}}{(d-1)^2-1}\right]- 4y\frac{d}{d-2}\left[\frac{(d-1)^{2}}{(d-1)^2-1}\right]\\
   &\le 4y\frac{d}{d-2}\left[\frac{(d-1)^{2i^*}}{(d-1)^2-1}\right]- 4y\frac{d}{d-2}\left[\frac{1}{(d-1)^2-1}\right] \qquad (\because (d-1)^2>1)\\
   &=\frac{1}{(d-1)^2-1}\left[4yd\frac{(d-1)^{2i^*}-1}{d-2}\right]\\
   &= \frac{1}{(d-1)^2-1}S_{i^*}\\
%\Rightarrow \alpha\le \frac{1}{(d-1)^2-1}S_{i^*}
\end{split}
\end{equation*}
\end{proof}

Since $\frac{1}{(d-1)^2-1}\le \frac{1}{3}$ for $d\ge 3$, Lemmas $\ref{lem-01}$ and $\ref{lem-04}$ imply

\begin{corollary}
\label{cor-02}
$\alpha\le \frac{S_{i^*}}{3}$
\end{corollary}

Denote by $A_1$ the agent that starts earlier, and by $A_2$ the agent that starts later. (In the case of simultaneous start, names $A_1$ and $A_2$ are given arbitrarily). 
Let $\delta \geq 0$ denote the delay of the start of  $A_2$ w.r.t the start of $A_1$.

The next lemma shows that if the delay is sufficiently large then agents meet during the critical stage of the earlier agent.

\begin{lemma}
\label{lem-05}
Let $\alpha$ be the time in which agent $A_1$ reaches its critical stage.
 If $\delta\ge \alpha+3a(r_i^*)$, then agents meet during the critical stage of $A_1$ .  
\end{lemma}

\begin{proof}
Agent $A_1$ takes time $\alpha$ to reach its critical stage. Since each transformed label starts with bits 01, during its critical stage agent $A_1$ first executes two consecutive passivity cycles, i.e., for a time $2a(r_i^*)$ it does not moves and then it starts executing its two consecutive activity cycles and each activity cycle takes time $a(r_i^*)$. Hence, in time at most $\alpha+3a(r_i^*)$ since its start agent $A_1$ reaches the node initially occupied by agent $A_2$. Thus, if $\delta \ge \alpha+3a(r_i^*)$, i.e., if  agent $A_2$ remains inactive till that time, agent $A_1$ meets $A_2$ at its initial position. This happens during the critical stage of $A_1$.
\end{proof}

\subsubsection{Labels of equal length}

In this section we assume that the transformed labels of the agents have the same length. This implies that the duration of each stage $i$ is the same for both agents.
We denote by $T_1(i)$ and $T_2(i)$ the time when $A_1$ (resp. $A_2$) starts its stage $i$.
\begin{lemma}
\label{lem-06}
If $\delta\le \alpha +3a(r_i^*)$ then agent $A_2$ starts its critical stage  before the end of the critical stage of $A_1$. 
\end{lemma}
\begin{proof}
 Let $y$ be the length of the transformed labels of the agents. The time taken by agent $A_2$ to reach its critical stage  is $\delta+\alpha$ after the start of agent $A_1$, and the time taken by agent $A_1$ to reach its stage $i^*+1$ is $\alpha+S_{i^*}$. We prove the lemma by contradiction. Suppose that agent $A_2$ starts its critical stage after the end of the critical stage of $A_1$. Hence
 $\delta+\alpha > \alpha+S_{i^*}$ which implies $\delta > S_{i^*}$, and thus  $\alpha+3a(r_{i^*})> 2a(r_{i^*})y$ because $ \delta\le \alpha +3a(r_i^*)$. Since $\alpha\le\frac{S_{i^*}}{3}$ and $\pi_{i^*}=2a(r_{i^*})$, we have $\frac{S_{i^*}}{3}+\frac{3}{2}\pi_{i^*}>S_{i^*}$ and thus 
 $\frac{\pi_{i^*}y}{3}+\frac{3}{2}\pi_{i^*}>\pi_{i^*}y$. This is a contradiction because $y\ge6$.
\end{proof}

\begin{lemma}
\label{lem-07}
If agents have labels of equal length, then they meet before the end of stage $i^*+1$ of agent $A_1$.
\end{lemma}
\begin{proof}
%First consider the case when $\delta=0$, i.e.,  when agents start the execution of the algorithm at the same time. Since the transformed labels are different and are of the same length, there is at least one
%position at which they differ. Let $w$ be the index of the first bit at which they differ. Since agents have labels of equal length and they start at the same time, the agents  start the execution of the $w$-th bit of their critical stage at the same time.
% The $w$-th bit of one agent is $0$ and of the other agent is $1$. The agent having $0$ does not move for time $2a(r_i^*)$ during this execution while the agent having $1$ completes during this time two activity cycles. Thus the agent having  the $w$-th bit 1 meets the other agent during the first activity cycle of this execution.
%Now consider the case when agents start at different times i.e., $\delta>0$. 
Let $\alpha$ be the time taken by $A_1$ to reach its critical stage from its starting time. If $\delta\ge \alpha+3a(r_i^*)$, then by Lemma $\ref{lem-04}$, agent $A_1$ meets agent $A_2$ in stage  $i^*$ of $A_1$, before $A_2$ starts its execution. Thus  we may assume that $\delta<\alpha+3a(r_i^*)$. Let the transformed label of $A_1$ be $b^1_1, b^1_2,\ldots,b^1_y$ and that of $A_2$ be $b^2_1, b^2_2,\ldots,b^2_y$. Our arguments to prove the lemma depend on the value of $\delta$ as follows:
\begin{itemize}
\item $\boldsymbol{\delta \le a(r_{i^*})}:$  Consider the smallest $j$ such that $b^1_j=0$ and $b^2_j=1$. (Since transformed labels of agents are different and have equal lengths, there are at least three such indices).  Consider the execution of the critical stage by the agents (Figure $\ref{same-1}$(A)). Since $\delta\le a(r_{i^*})$, agent $A_2$ fully executes one activity cycle corresponding to $b^2_j=1$ within the two passivity cycles of $A_1$ corresponding to $b^1_j=0$. Thus agent $A_2$ meets $A_1$ during the execution of the critical stage of $A_1$.
 
\item $\boldsymbol{a(r_{i^*})<\delta\le 3a(r_{i^*})}:$ The duration of execution of the first two bits 01 in the critical stage is $4a(r_{i^*})$ for both agents (Figure $\ref{same-1}$(B)). Since $a(r_{i^*})<\delta\le 3a(r_{i^*})$, 
agent $A_2$ is idle while executing its bit $b^2_1=0$ during a 
complete activity cycle of $A_1$ corresponding to $b^1_2=1$. Hence $A_1$ meets $A_2$ during the execution of the critical stage of $A_1$.

\begin{figure}[h]
   \vspace*{-.194in}
     \centering
    \includegraphics[scale = .45]{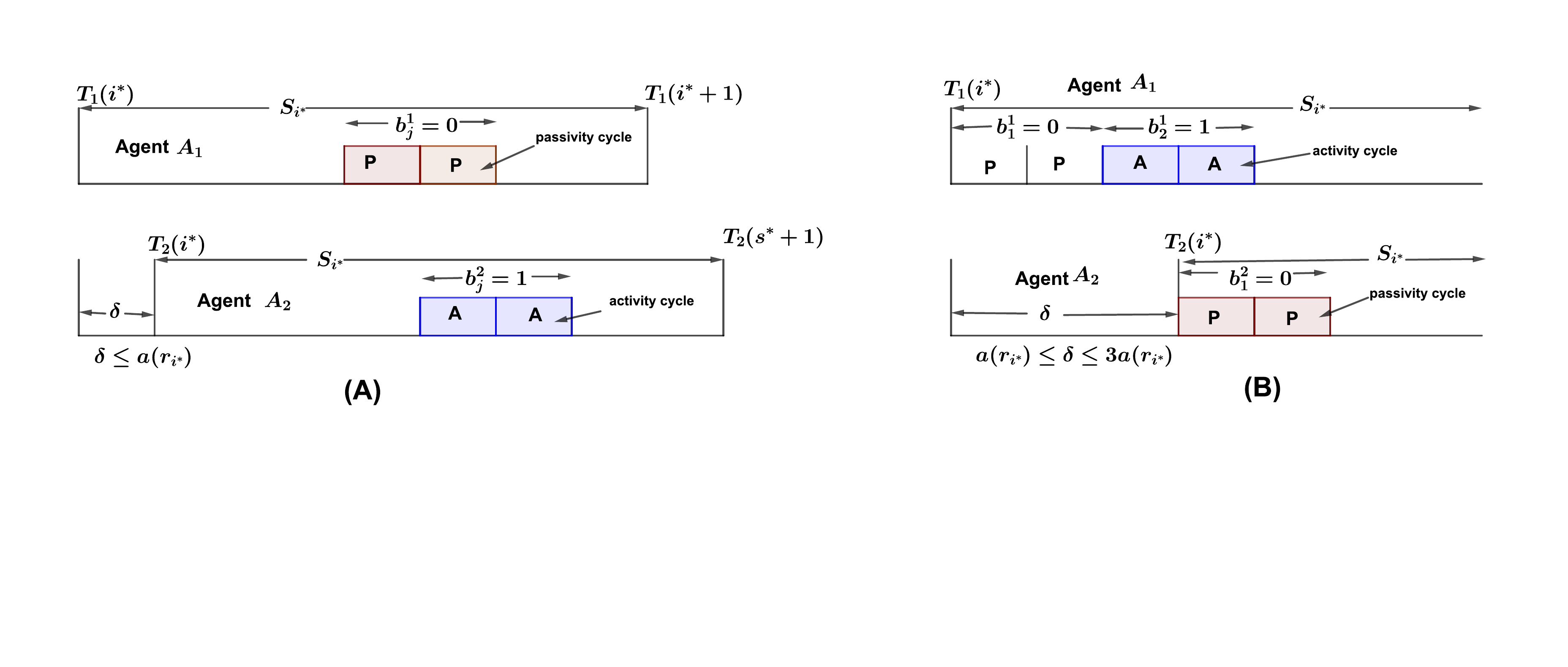}
   \vspace*{-1.2in}
    \caption{An illustration for the proof of Lemma $\ref{lem-07}$:  (A) when $\delta \le a(r_i^*)$, agent $A_2$ meets $A_1$ during the execution of the critical stage of $A_1$, and (B) when $a(r_i^*)<\delta\le 3a(r_i^*)$, agent $A_1$ meets $A_2$ during the execution of the critical stage of $A_1$.}
    \label{same-1}
\end{figure}

\item $\boldsymbol{3a(r_{i^*})<\delta\le 8a(r_{i^*})}:$  In this case at least a time segment of length $3a(r_{i^*})$ of the critical stage of $A_2$ is executed during the execution of stage $i^*+1$ of $A_1$ (Figure $\ref{same-2}$(A)).  The first time segment of length  $8a(r_{i^*})$ of stage $i^*+1$ of $A_1$ is devoted to the execution of the first bit of the transformed label, which is 0. Thus $A_1$ is idle during this time segment. The final time segment of length $4a(r_{i^*})$ of the critical stage of $A_2$  contains exactly two activity cycles of $A_2$ (since among the last two bits of the transformed label of each agent there is one bit $1$). Since $3a(r_{i^*})^*<\delta\le8a(r_{i^*})$,  at least one activity cycle of $A_2$ in its critical stage is completely executed during the first passivity period of $A_1$ in its stage $i^*+1$. Thus agent $A_2$ meets $A_1$ during stage $i^*+1$ of the latter. 
  
\item $\boldsymbol{8a(r_{i^*})<\delta< \alpha +3a(r_i)}:$  In this case, by Lemma $\ref{lem-05}$, agent $A_2$ starts its critical stage before the the end of the critical stage of $A_1$. 
Let $I$ be the initial time segment of length $8a(r_{i^*})$ of stage $i^*+1$ of $A_1$ (Figure $\ref{same-2}$(B)).
Since the duration of the execution of each bit in stage ${i^*+1}$ is at least $8a(r_{i^*})$ (by Corollary $\ref{cor-01}$), we know that 
during time segment $I$, agent $A_1$ executes the first bit of its transformed label, i.e., bit 0.
Hence during time segment $I$, agent $A_1$ is idle.
Since $\delta>8a(r_{i^*})$, during time segment $I$ agent $A_2$ still executes its critical stage. 
Since the duration of execution of each bit in the critical stage of the agents is $2a(r_{i^*})$, during time segment $I$ agent $A_2$ must perform complete executions of at least 3 consecutive bits of its transformed label. Among any three consecutive bits of any transformed label there must be at least one bit 1. Hence, during time segment $I$ in which agent $A_1$ is idle, agent $A_2$ performs two activity cycles of its critical stage. Thus it must meet agent $A_1$ before the end of stage  ${i^*+1}$ of the latter.

%Since the duration of execution of each bit in  the critical stages of the agents is $2a(r_{i^*})$ and $\delta>8a(r_{i^*})$, agent $A_2$ executes at least $4$ bits of its critical stage  during the execution of stage $i^*+1$ of $A_1$. Since the duration of the execution of each bit in stage ${i^*+1}$ is at least $8a(r_{i^*})$ (by Corollary $\ref{cor-01}$), agent $A_2$ can execute exactly $4$ bits of execution of stage $S_{i^*}(A_2)$ during the execution of first $0$ in stage $S_{i^*+1}(A_1)$ of $A_1$. Since any $4$ consecutive  bits of transformed labels contain at least one $1$, agent $A_2$ performs two activity cycles of stage $S_{i^*}(A_2)$ during the first passivity cycle of stage $S_{i^*+1}(A_1)$. Thus, agent $A_2$ catches agent $A_1$ in stage at most $S_{i^*+1}(A_1)$.
%This completes the proof.
\end{itemize}
\end{proof}
  \begin{figure}[h]
   \vspace*{-.194in}
     \centering
    \includegraphics[scale = .45]{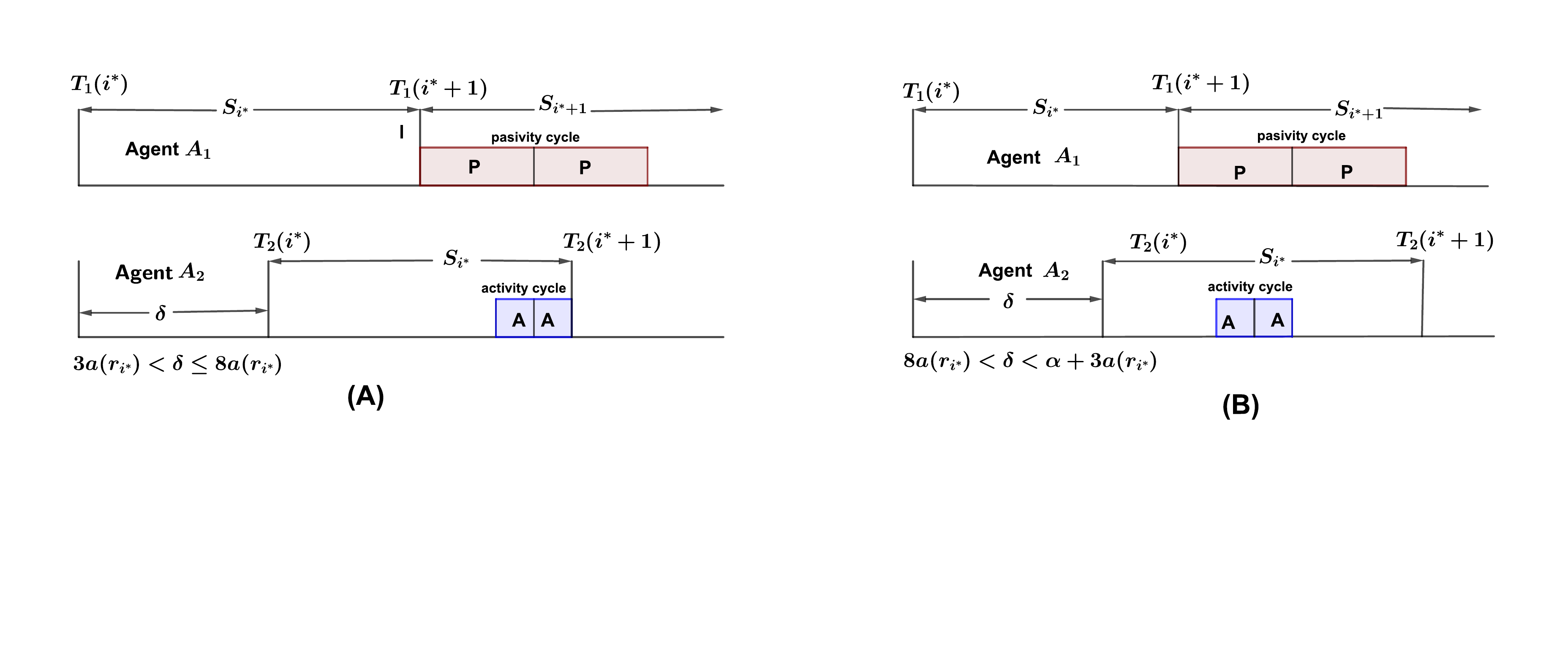}
   \vspace*{-1.2in}
    \caption{An illustration for the proof of Lemma $\ref{lem-07}$:  (A) when $3a(r_i^*)<\delta\le 8a(r_i^*)$, agent $A_2$ meets agent $A_1$ during the stage ${i^*+1}$ of the latter, and (B) when $8a(r_i^*)<\delta\le \alpha +3a(r_i)$, agent $A_1$ meets agent $A_2$ before the end of stage ${i^*+1}$ of $A_1$.}
    \label{same-2}
\end{figure}

\subsubsection{Labels of different lengths}

In this section we assume that the labels of the agents have different lengths.
Let $X$ denote the length of the shorter transformed label. We refer to the agent having this label as the {\it faster agent}, and denote it by $A_f$. Let $\beta X$, for $\beta >1$, be the length of the longer transformed label. We refer to the agent having this label as the {\it slower agent}, and denote it by $A_s$. These names are chosen due to the fact that, since the duration of any stage of an agent is proportional to the length of its transformed label, the agent with shorter transformed label completes its stages faster than the agent with longer transformed label. 
Since $\beta>1$ and $X\ge6$, we have $\beta X\ge X+6$. Let $S_i(f)$ and $S_i(s)$ denote the lengths of stage $i$ of the faster and slower agents, respectively. Let agent $A_s$ start its critical stage during the stage $i^*+p$ of agent $A_f$, where $p$ is an integer.

We will use the following technical lemma to estimate by which stage of $A_s$ the agents will meet.

\begin{lemma}
\label{lem-08}
Let agent $A_s$ start its stage $i^*+q$ during the stage $i^*+k$ of $A_f$ where $q\ge0$ and $k\ge q+1$. Let $I_s=S_{i^*+q+1}(s)+S_{i^*+q}(s)-S_{i^*+k}(f)$. If $I_s\ge 4\pi_{i^*+q+1}$, then agent $A_s$ meets agent $A_f$ in stage at most $i^*+k+1$ of $A_f$. Furthermore, during this meeting agent $A_s$ is in stage at most $i^*+q+1$.
\end{lemma}

\begin{proof}
Let $T_s(i)$ and $T_f(i)$ denote the starting times of the stage $i$ of the agents $A_s$ and $A_f$ respectively. Since agent $A_s$ starts its stage $i^*+q$ during the stage $i^*+k$ of $A_f$, we have $T_f(i^*+k)\le T_s(i^*+q)<T_f(i^*+k+1)$. We consider the following two cases:
\begin{itemize}
\item $\boldsymbol{T_s(i^*+q+1)<T_f(i^*+k+1):}$ In this case agent $A_s$ completes its stage $i^*+q$  and starts its stage $i^*+q+1$ before the end of stage $i^*+k$ of agent $A_f$ (Figure $\ref{diff-1}$(A)). Since $I_s\ge 4\pi_{i^*+q+1}$, we have $T_s(i^*+q+2)>T_f(i^*+k+1)$ and hence the part of stage $i^*+q+1$ of $A_s$ executed after $T_f(i^*+k+1)$ has length at least $4\pi_{i^*+q+1}$.
 Now since $k\ge q+1$, we have $\pi_{i+k+1}\ge 4\pi_{i^*+q+1}$ and this implies that agent $A_s$ fully executes at least  $3$ consecutive bits of stage $i^*+q+1$ during the execution of the first bit of stage $i^*+k+1$ of agent $A_f$. The execution of any $3$ consecutive bits of an agent contains at least two activity cycles and during the execution of the first bit of any stage, an agent executes two passivity cycles. Thus, agent $A_s$ executes at least two activity cycles of stage $i^*+q+1$ during the first two consecutive passivity cycles of stage $i^*+k+1$ of agent $A_f$. This implies that agent $A_s$ meets agent $A_f$ during the stage $i^*+k+1$ of the latter. During this meeting agent $A_s$ is in stage $i^*+q+1$.
\begin{figure}[h]
   \vspace*{-.194in}
     \centering
    \includegraphics[scale = .45]{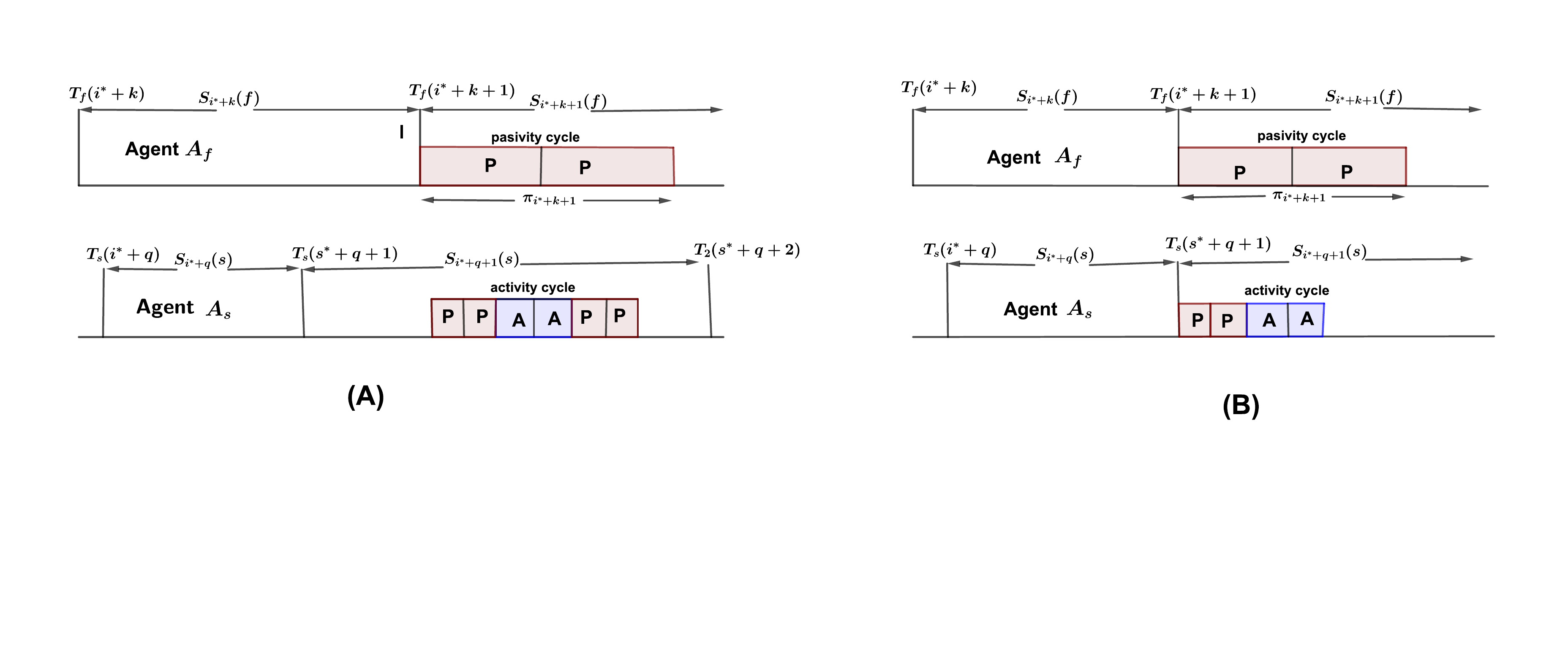}
   \vspace*{-1.2in}
    \caption{An illustration for the proof of Lemma $\ref{lem-08}$:  (A) $T_s(i^*+q+1)<T_f(i^*+k+1)$, and (B) $T_s(i^*+q+1)=T_f(i^*+k+1)$. In both cases, agent $A_s$ meets $A_f$ during the execution of  stage $i^*+k+1$ of $A_f$; during this meeting agent $A_s$ is in stage $i^*+q+1$.}
    \label{diff-1}
\end{figure}

\item $\boldsymbol{T_s(i^*+q+1)\ge T_f(i^*+k+1):}$ Here we consider the following two subcases:
\begin{itemize}
\item $\boldsymbol{T_s(i^*+q+1)=T_f(i^*+k+1):}$ Since $\pi_{i+k+1}\ge 4\pi_{i^*+q+1}$, agent $A_s$ executes at least the first 4 bits of its stage $i^*+q+1$ during the execution of the first bit 0 of stage $i^*+k+1$ of agent $A_f$ (Figure $\ref{diff-1}$(B)). This implies that agent $A_s$ meets agent $A_f$ during the execution of the first two passivity cycles of stage $i^*+k+1$ of agent $A_f$ corresponding to the execution of the first bit $0$. The agent $A_s$ is in stage $i^*+q+1$ during this meeting.

\item $\boldsymbol{T_s(i^*+q+1)>T_f(i^*+k+1):}$ In this case a portion of the stage $i^*+q$ of agent $A_s$ is executed during the stage $i^*+k+1$ of agent $A_f$. Let $I_s(i^*+q)$ and $I_s(i^*+q+1)$ be the durations of the portions of the stages $i^*+q$ and $i^*+q+1$ of $A_s$ executed during the stage $i^*+k+1$ of $A_f$. Note that $I_s(i^*+q+1)$ is of length zero when $T_s(i^*+q+1)\ge T_f(i^*+k+2)$. Let $I$ denote the time segment of stage $i^*+k+1$ of $A_f$ during which agent $A_f$ executes its first bit $0$. Since $k\ge q+1$, we have $I=\pi_{i^*+k+1}\ge 4\pi_{i^*+q+1}$ (by Corollary $\ref{cor-01}$). The execution of the last two bits of stage $i^*+q$ is of duration $2\pi_{i^*+q}$ and exactly one of these last two bits is 1. This implies that the time segment of length $\frac{3}{2}\pi_{i^*+q}$ at the end of stage $i^*+q$ of agent $A_s$ contains at least one activity cycle. Thus, if $I_s(i^*+q)\ge\frac{3}{2}\pi_{i^*+q}$ (Figure $\ref{diff-2}$(A)), then agent $A_s$ executes at least one activity cycle of stage $i^*+q$ during the time segment $I$ (since $I\ge 4\pi_{i^*+q+1}\ge 4^2\pi_{i^*+q}$). Since during  the whole execution of $I$ agent $A_f$ remains idle, agent $A_s$ meets agent $A_f$ in stage $i^*+k+1$ of the latter. During this meeting agent $A_s$ is in stage at most $i^*+q$. Next suppose $I_s(i^*+q)<\frac{3}{2}\pi_{i^*+q}$ (Figure $\ref{diff-2}$(B)). Note that in this case, agent $A_s$ starts its stage $i^*+q+1$ during the execution of the first bit $0$ of stage $i^*+k+1$ of agent $A_f$. Now, 
\begin{equation*}
\begin{split}
I-I_s(i^*+q) &\ge 4\pi_{i^*+q+1}-\frac{3}{2}\pi_{i^*+q}  \qquad ( \because I=\pi_{i^*+k+1} \ge 4\pi_{i^*+q+1}) \\
 & \ge3\pi_{i^*+q+1}
\end{split}
\end{equation*} 

This implies that agent $A_s$ fully executes at least the first 3 bits of stage $i^*+q+1$ during time segment $I$. Since the execution of the first two bits in any stage contains exactly two activity cycles, agent $A_s$ executes two activity cycles of stage $i^*+q+1$ during $I$ and it meets agent $A_f$ in stage $i^*+k+1$ of the latter. During this meeting agent $A_s$ is in stage at most $i^*+q+1$.
\begin{figure}[h]
   \vspace*{-.194in}
     \centering
    \includegraphics[scale = .45]{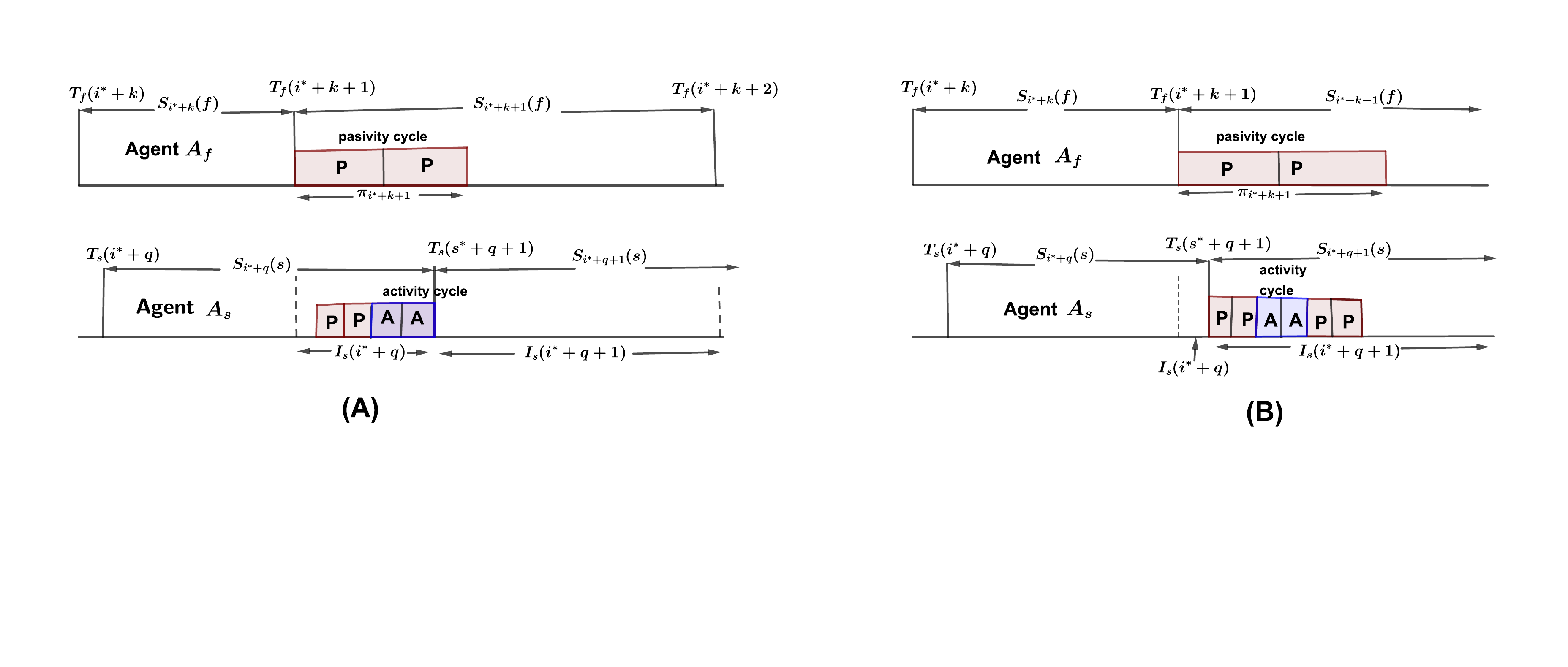}
   \vspace*{-1.2in}
    \caption{An illustration for the proof of Lemma $\ref{lem-08}$: (A) $I_s(i^*+q)\ge\frac{3}{2}\pi_{i^*+q}$, and (B) $I_s(i^*+q)<\frac{3}{2}\pi_{i^*+q}$. In both cases, agent $A_s$ meets $A_f$ during the execution of stage $i^*+k+1$ of $A_f$.  In case (A), agent $A_s$ is in stage at most $i^*+q$ and in case (B) agent $A_s$ is in stage at most $i^*+q+1$.}
    \label{diff-2}
\end{figure}

\end{itemize}
\end{itemize} 
\end{proof}

We now proceed to the proof that agents meet always  by the end of stage $i^*+2$ of $A_s$. The proof is split into two cases in Lemmas \ref{lem-09} and \ref{lem-11}.

\begin{lemma}
\label{lem-09}
If agent $A_f$ starts its execution before the start of agent $A_s$, then the agents meet during stage at most $i^*+2$ of $A_s$.
\end{lemma}
\begin{proof}
In this case, we have $p\ge 0$. We consider two cases.
\begin{itemize}
\item $\boldsymbol{p=0}:$ In this case, agent $A_s$ starts its critical stage during the execution of the critical stage of $A_f$ (Figure $\ref{diff-3}(A)$). We have $S_{i^*}(f)=\pi_{i^*}X$ and $S_{i^*}(s)=\pi_{i^*}\beta X$. Since $\beta >1$, we have $S_{i^*}(f)< S_{i^*}(s)$. 

\begin{equation*}
\begin{split}
S_{i^*}(s)-S_{i^*}(f) & =\pi_{i^*}\beta X-\pi_{i^*}X \\
      &\ge \pi_{i^*}(X+6)-\pi_{i^*}X \qquad( \because \beta X\ge X+6)   \\
 & = 12a(r_{i^*})
\end{split}
\end{equation*}
Since the execution of one bit of stage  $i^*$ has duration $2a(r_{i^*})$, this implies that at least $5$ bits of stage $i^*$ of agent $A_s$ are fully executed during the initial part of the execution of stage $i^*+1$ of agent $A_f$. Thus, since $\pi_{i^*+1}\ge 4\pi_{i^*}$, agent $A_s$ fully executes at least  $3$ consecutive bits of stage $i^*$ during the execution of the first bit of stage $i^*+1$ of agent $A_f$. The execution of any $3$ consecutive bits of an agent contains at least two activity cycles and during the execution of the first bit of any stage, an agent executes two passivity cycles. Thus, agent $A_s$ executes at least two activity cycles of stage $i^*$ during the first two consecutive passivity cycles of stage $i^*+1$ of agent $A_f$. This implies that agent $A_s$ meets agent $A_f$ during the stage $i^*+1$ of the latter. 

\begin{figure}[h]
   \vspace*{-.194in}
     \centering
    \includegraphics[scale = .45]{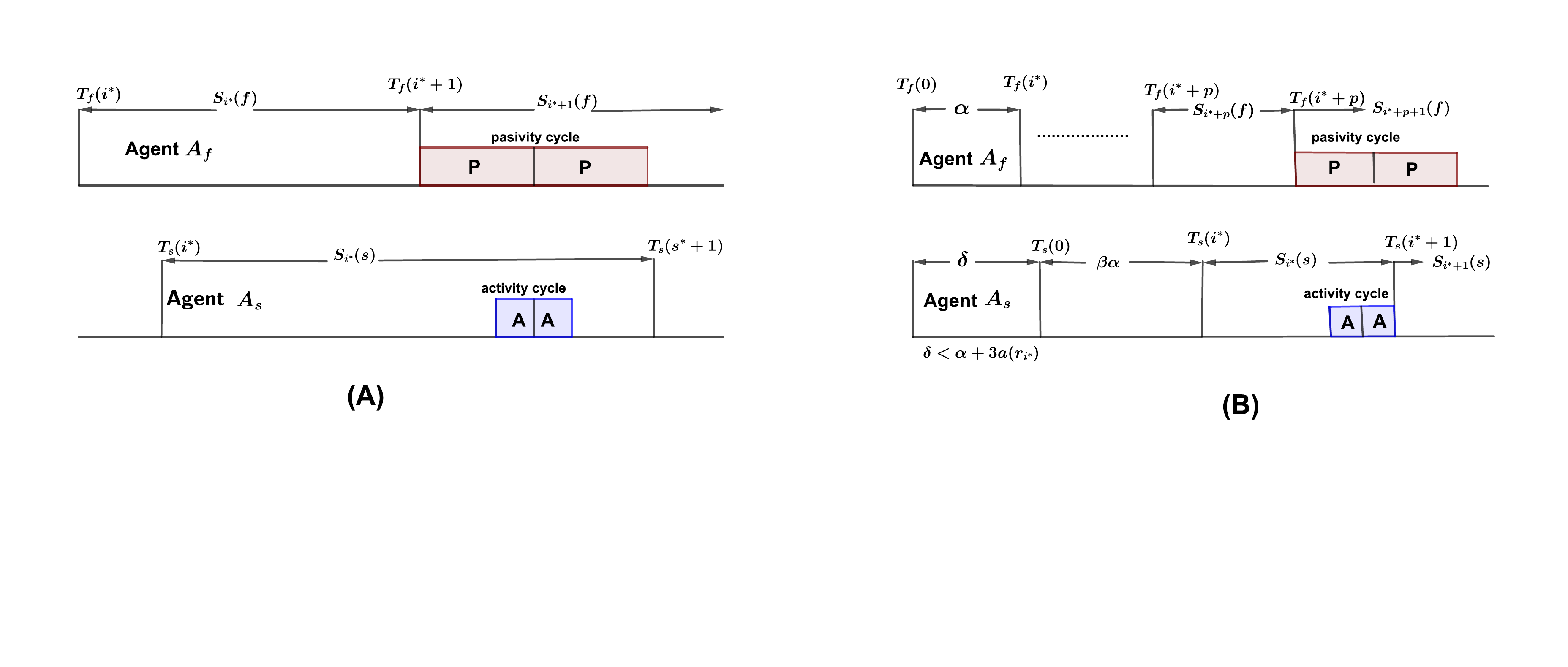}
   \vspace*{-1.2in}
    \caption{An illustration for the proof of Lemma $\ref{lem-09}$:  (A) when $p=0$, agent $A_s$ meets $A_f$ during the execution of stage $i^*+1$ of $A_f$, and (B) when $p\ge 1$, agent $A_s$ meets $A_f$ during the execution of stage $i^*+p+1$ of $A_f$.}
    \label{diff-3}
\end{figure}
\item $\boldsymbol{p\ge1}:$ Let $\alpha$ be the time in which agent $A_f$ reaches its critical stage. We assume $\delta\le \alpha+\frac{3}{2}\pi_{i^*}$ (otherwise, by Lemma $\ref{lem-05}$, agents meet in stage $i^*$ of agent  $A_f$). Since agent $A_s$ starts its critical stage during the execution of stage $i^*+p$ of agent $A_f$ (Figure $\ref{diff-3}(B)$), we have,
\begin{alignat}{2}
  &\quad
&\delta +\beta \alpha
&\ge \alpha +S_{i^*}(f)+S_{i^*+1}(f)+\cdots+S_{i^*+p-1}(f)\notag\\
&
&\alpha +\frac{3}{2}\pi_{i^*}+\beta\alpha
&\ge \alpha + S_{i^*+p-1}(f)\notag\qquad(\because \delta<\alpha+\frac{3}{2}\pi_{i^*})\\
&
&\frac{3}{2}\pi_{i^*}+\beta \frac{S_{i^*}(f)}{3} 
&\ge S_{i^*+p-1}(f) \qquad(\because \alpha<\frac{S_{i^*}(f)}{3})
\end{alignat}
%
%
%From (1), we have
%\begin{alignat}{2}
%  &\Rightarrow\quad
%&\delta+\beta\alpha
%&\ge 4^p\alpha\notag\\
%&\Rightarrow
%&\alpha+16r_{i^*}+\beta\alpha
%&\ge 4^p\alpha\notag\\
%&\Rightarrow
%& +\beta 8DX\ge 4^p8DX \qquad(\because \alpha=\frac{8DX}{3})
%%\implies 3\delta\ge 4^p8DX-\beta 8DX$$
%%$$\implies 3\alpha+48D\ge 4^p8DX-\beta 8DX \qquad(\because \delta\le \alpha+16D)$$
%%$$\implies 8DX+48D\ge 4^p8DX-\beta 8DX \qquad(\because \alpha<\frac{8DX}{3}) \qquad\ldots \qquad(2)$$
%\end{alignat}

%Next we ask the following question: Is it possible that during the whole execution of stage $S_{i^*+p}(f)$, agent $A_s$ executes its stages $S_{i^*}(s)$ and $S_{i^*+1}(s)$? The answer is no which follows from the following. If the stages $S_{i^*}(s)$ and $S_{i^*+1}(s)$ are executed within the duration of the stage $S_{i^*+p}(f)$, then we must have the following inequality
%\begin{alignat}{2}
% &\Rightarrow\quad
% &4^p8r_{i^*}X
% &\ge \beta 8r_{i^*}X+\beta 32r_{i^*}X\notag\\
% &\Rightarrow
%&4^p8r_{i^*}X-\beta 8r_{i^*}X
%&\ge\beta 32r_{i^*}X\notag\\
%&\Rightarrow
%&8r_{i^*}X+48r_{i^*}
%&\ge\beta 32r_{i^*}X\qquad (by (1))\notag\\
%&\Rightarrow
%& X+6
%&\ge \beta 4X\notag\\
%&\Rightarrow
%&6
%&\ge (4\beta-1)X\notag
%\end{alignat}
%This is not possible since $\beta\ge3$ and $X\ge 6$.
%
%Now, the length of the stage $S_{i^*+p}(f)$ is $4^p8r_{i^*}X$ and the lengths of stages $S_{i^*}(s)$ and $S_{i^*+1}(s)$  are $\beta 8r_{i^*}X$ and $\beta 32r_{i^*}X$ respectively.
%We compute the overflow between the stage $i^*+p+1$ of $A_f$ and stages $i^*+1$ and $i^*$ of agent $A_s$.
 Let $I_s=S_{i^*+1}(s)+S_{i^*}(s)-S_{i^*+p}(f)$. We compute a lower bound on $I_s$ as follows:
\begin{equation*}
\begin{split}
I_s=S_{i^*+1}(s)+S_{i^*}(s)-S_{i+p}(f) &\ge  S_{i^*+1}(s)+S_{i^*}(s)-5S_{i+p-1}(f) \qquad \text{(by Lemma \ref{lem-03})} \\
 & \ge S_{i^*+1}(s)+S_{i^*}(s)-\frac{15}{2}\pi_{i^*}- \frac{5}{3}\beta S_{i^*}(f)\qquad \text{(by (3))}\\
 &\ge \pi_{i^*+1}\beta X+\pi_{i^*}\beta X-\frac{5}{3}\pi_{i^*}\beta X-\frac{15}{2}\pi_{i^*}\\
 &\ge \frac{10}{3}\pi_{i^*}\beta X-\frac{15}{2}\pi_{i^*} \qquad (\text{by Corollary \ref{cor-01}},\pi_{i^*+1} X\ge 4\pi_{i^*} X )\\
 &\ge \frac{10}{3}\pi_{i^*}(X+6)-\frac{15}{2}\pi_{i^*} \qquad (\because \beta X\ge X+6)\\
 &\ge \frac{10}{3}\pi_{i^*}X \\
 &\ge \frac{2}{3}\pi_{i^*+1}X \qquad (\text{by Lemma \ref{lem-03}}, \pi_{i^*+1} X\le 5\pi_{i^*} X)\\
 &\ge 4\pi_{i^*+1} \qquad (\because X\ge 6)\\
\end{split}
\end{equation*}  

Hence the lemma is true in this case by Lemma $\ref{lem-08}$ substituting $q=0$ and $k=p$.

%
%Here, the critical stage of $A_s$ starts during the stage $i^*+1$ of $A_f$. We consider $S_{i^*}(s)$, $S_{i^*+1}(s)$ and $S_{i^*+1}(f)$. Let $I_s=S_{i^*+1}(s)+S_{i^*}(s)-S_{i^*+1}(f)$. Then
%\begin{equation*}
%\begin{split}
%I_s &= S_{i^*+1}(s)+S_{i^*}(s)-S_{i^*+1}(f)\\
% & =\pi_{i^*+1}\beta X+\pi_{i^*}\beta X-\pi_{i^*+1}X \\
%      &\ge\pi_{i^*+1}(X+6)+\pi_{i^*}(X+6)-\pi_{i^*+1}X  \qquad ( \because \beta X\ge X+6)\\
%      &= 6\pi_{i^*+1}+\pi_{i^*}(X+6)  \\
% & \ge 6\pi_{i^*+1}
%\end{split}
%\end{equation*} 

\end{itemize}
%
%Now we compute  a lower bound of $DX(32\beta-8)-48D$ as follows
%\begin{alignat}{2}
% &\Rightarrow\quad
%&r_{i^*}X(32\beta-8)-48r_{i^*}
%&\ge r_{i^*}X(96-8)-48r_{i^*} \qquad (\because \beta \ge3)\notag\\
%&\Rightarrow
%& r_{i^*}X(32\beta-8)-48r_{i^*}
%&\ge 88r_{i^*}X-48r_{i^*} \qquad\notag\\
%&\Rightarrow
%& DX(32\beta-8)-48D
%&\ge 528r_{i^*}-48r_{i^*}\qquad (\because X\ge6)\notag\\
%&\Rightarrow
%&r_{i^*}X(32\beta-8)-48r_{i^*}
%&\ge 480r_{i^*}D\notag
%\end{alignat}
%Thus from $(2)$, we have
%$$32\beta r_{i^*}X+8\beta r_{i^*}X-4^p8r_{i^*}X> 128r_{i^*}$$
%This implies that at least 4 bits of execution stage $S_{i^*+1}(s)$ is performed during the execution of 1st bit of stage $S_{i^*+p+1}(f)$ (each bit execution in stage $S_{i^*+p+1}(f)$ takes $4^{p+1}r_{i^*}$ time and since $p\ge1$, this time is at least $128r_{i^*}$). The 1st bit of stage $S_{i^*+p+1}(f)$ is $0$ and thus meeting is achieved during this phase. 
\end{proof}

\begin{lemma}
\label{lem-10}
If agent $A_s$ starts its execution before the start of agent $A_f$ then
agent $A_s$ starts its stage $i^*+1$ after the start of stage $i^*$ of $A_f$.
\end{lemma}
\begin{proof}
If $p\ge0$, then the lemma is obvious by the definition of $p$. Thus we assume $p<0$. Hence the agent $A_s$ starts its critical stage before the start of the critical stage of agent $A_f$ (Figure $\ref{diff-4}(A)$). We know that $\delta\le \beta\alpha+\frac{3}{2}\pi_{i^*}$, where $\alpha$ is the time in which agent $A_f$ reaches its critical stage (otherwise agent $A_s$ meets agent $A_f$ before the latter starts its execution).

\begin{figure}[h]
   \vspace*{-.194in}
     \centering
    \includegraphics[scale = .42]{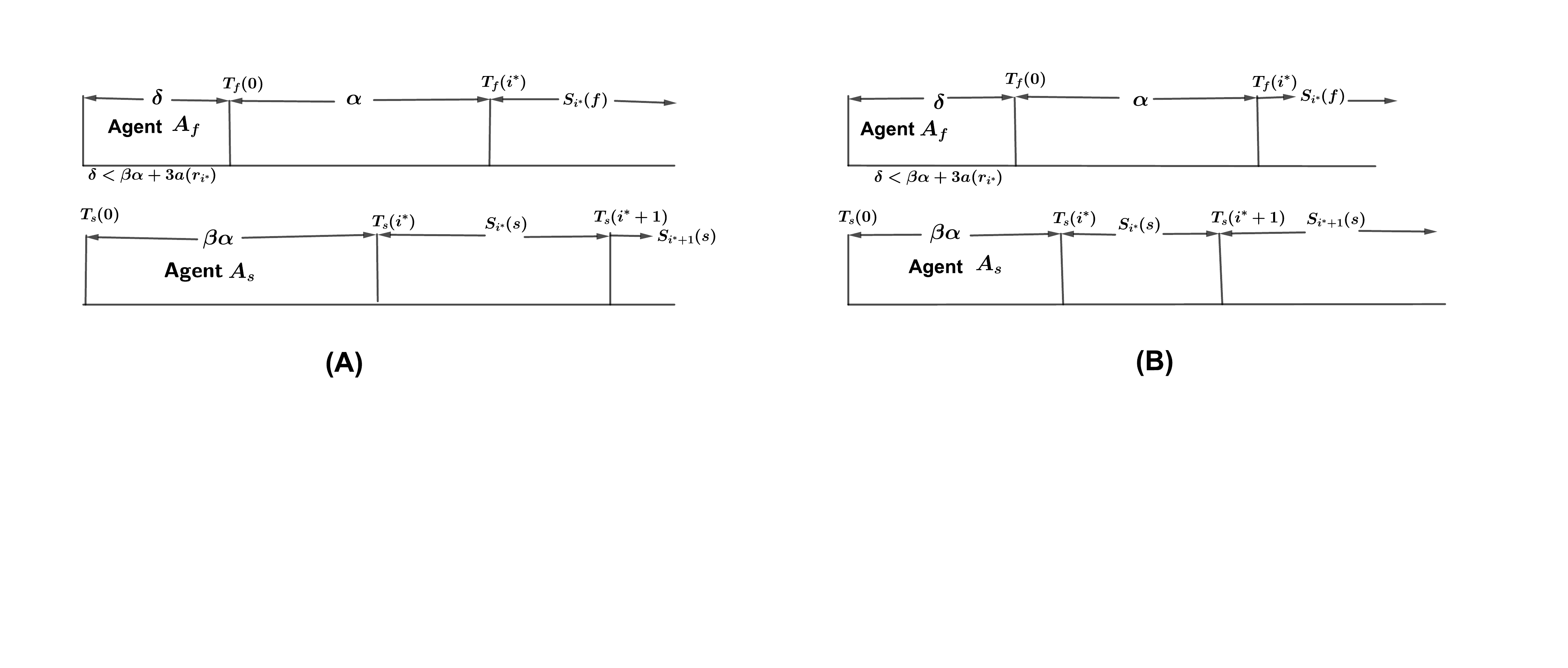}
   \vspace*{-1.2in}
    \caption{An illustration for the proof of Lemma $\ref{lem-10}$:  (A) when agent $A_s$ starts its stage $i^*+1$ after the start of stage $i^*$ of agent $A_f$, and (B) when agent $A_s$ starts its stage $i^*+1$ before the start of stage $i^*$ of agent $A_f$.}
    \label{diff-4}
\end{figure}

Suppose for contradiction that agent $A_s$ starts its stage $i^*+1$ before the start of stage $i^*$ of $A_f$ (Figure $\ref{diff-4}(B)$). Then we have the following inequalities:
\begin{alignat}{2}
 &\quad
&\delta+\alpha
&\ge  \beta\alpha+S_{i^*}(s)\notag\\
&
& \beta\alpha+\frac{3}{2}\pi_{i^*}+\alpha
&\ge  \beta\alpha+S_{i^*}(s) \qquad\notag\\
&
& \frac{3}{2}\pi_{i^*}+\frac{S_{i^*}(f)}{3}
&\ge S_{i^*}(s)\qquad (\because \alpha<\frac{S_{i^*}(f)}{3})\notag\\
&
&\frac{3}{2}\pi_{i^*}+\frac{\pi_{i^*}X}{3}
&\ge \beta \pi_{i^*}X\notag\\
&
&\frac{3}{2}\pi_{i^*}+\frac{\pi_{i^*}X}{3}
&\ge  \pi_{i^*}(X+6)\qquad (\because \beta X\ge X+6)\notag
\end{alignat}

This is  a contradiction, since $X\ge 6$). Hence the lemma is true.
%
%$$\delta+\alpha\ge \beta\alpha+\beta 8r_{i^*}X$$
%$$\implies \beta\alpha+16r_{i^*}+\alpha\ge\beta\alpha+\beta 8 r_{i^*}X$$
%$$\implies 16r_{i^*}+\alpha\ge 8r_{i^*}X+48r_{i^*}$$
%$$\implies \frac{8r_{i^*}X}{3}-\frac{8X}{3}\ge 8r_{i^*}X+32r_{i^*}$$
%$$\implies \frac{8r_{i^*}X}{3}\ge 8r_{i^*}X+32r_{i^*}+\frac{8X}{3}$$
%which is not possible. Hence we have the lemma.  
\end{proof}

\begin{lemma}
\label{lem-11}
If agent $A_s$ starts its execution before the start of agent $A_f$, then agents meet during stage at most $i^*+2$ of $A_s$.
\end{lemma}
\begin{proof}
Let $\delta>0$  denote the delay of the start of $A_f$ w.r.t the start of $A_s$, and let $\alpha$ denote the time in which agent $A_f$ reaches its critical stage. We assume $\delta<\beta\alpha+3a(r_{i^*})$ (otherwise agent $A_s$ meets agent $A_f$ during the stage $i^*$ of the former). Let agent $A_s$ start its critical stage during stage $i^*+p$ of $A_f$, where $p$ is an integer. Note that in this case $p$ can assume either a non-negative or a negative value. We consider two cases.

\begin{itemize}
\item $\boldsymbol{p\ge 0:}$ In this case the proof of the lemma is similar to the proof of Lemma $\ref{lem-09}$. Indeed, if $p=0$, the proof is exactly the same (Figure $\ref{diff-5a}(A)$). Now consider the case when $p\ge 1$(Figure $\ref{diff-5a}(B)$). We only show that inequality $(3)$ also holds in this case and the rest of the proof is the same. Since the critical stage of $A_s$ starts during the stage $i^*+p$ of agent $A_f$ we have
\begin{alignat}{2}
  &\quad
&\beta \alpha
&\ge \delta +\alpha +S_{i^*}(f)+S_{i^*+1}(f)+\cdots+S_{i^*+p-1}(f)\notag\\
&
&\frac{3}{2}\pi_{i^*}+\beta \frac{S_{i^*}(f)}{3} 
&\ge S_{i^*+p-1} \qquad(\because \alpha<\frac{S_{i^*}(f)}{3} \hspace*{.1in}\text{and adding} \hspace*{.1in} \frac{3}{2}\pi_{i^*}\ge 0 \hspace*{.1in} \text{on the left hand side})\notag
\end{alignat}
Hence inequality $(3)$ holds in this case.

\begin{figure}[h]
   \vspace*{-.194in}
     \centering
    \includegraphics[scale =.42]{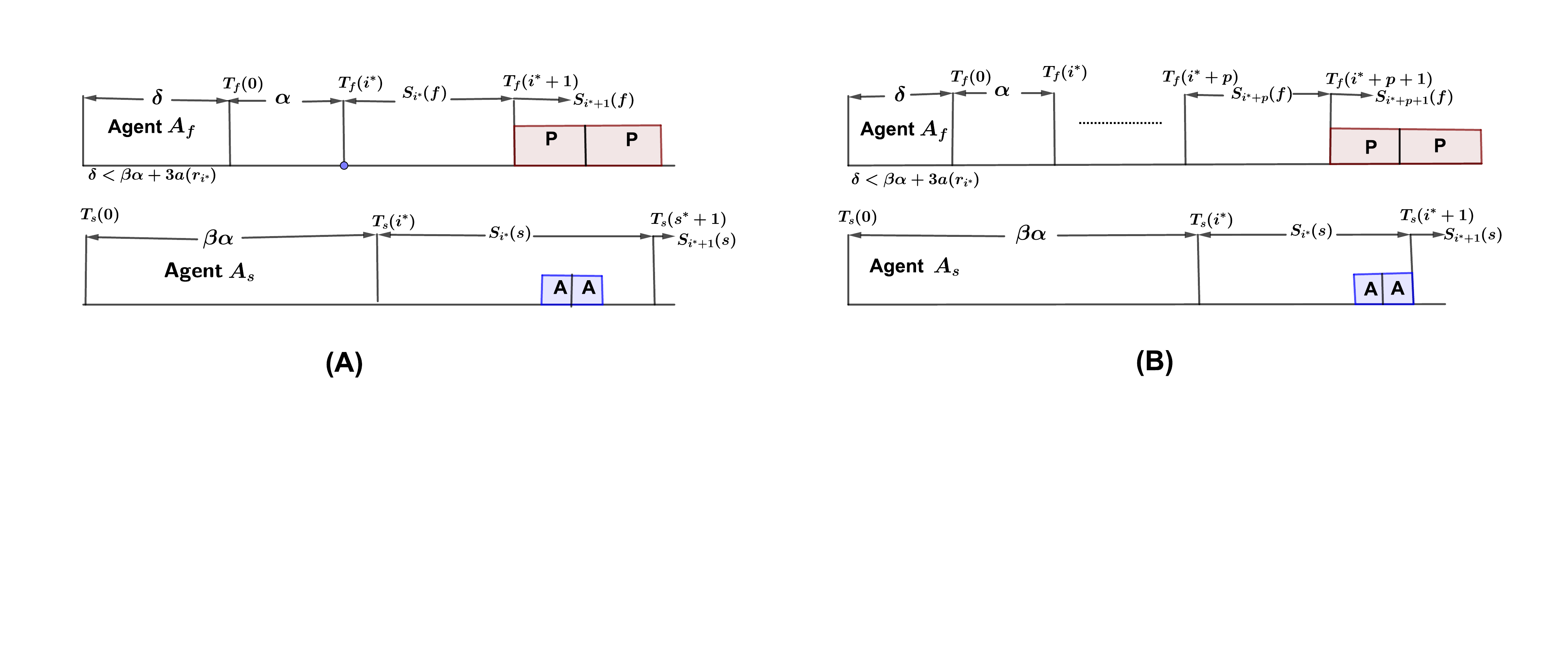}
   \vspace*{-1.2in}
    \caption{An illustration for the proof of Lemma $\ref{lem-11}$ when $p\ge 0$ and:  (A) agent $A_s$ starts its stage $i^*$ during the stage stage $i^*$ of agent $A_f$, and (B) agent $A_s$ starts its stage $i^*$ during the stage $i^*+p$ of agent $A_f$ for $p\ge1$.}
    \label{diff-5a}
\end{figure}

\item $\boldsymbol{p<0:}$ In this case the critical stage of $A_s$ starts before the start of the critical stage of $A_f$. By Lemma $\ref{lem-10}$, agent $A_s$ starts its stage $i^*+1$ after the start of stage $i^*$ of $A_f$. Let $T_s(i)$ and $T_f(i)$ denote the starting times of the stage of $i$ of the agents $A_s$ and $A_f$ respectively. We consider two subcases: $T_s(i^*+1)<T_f(i^*+1)$ and $T_s(i^*+1)\ge T_f(i^*+1)$.

\begin{itemize}
\item $\boldsymbol{T_s(i^*+1)\ge T_f(i^*+1):}$ Let agent $A_s$ start its stage $i^*+1$ during the stage $i^*+1+u$ of agent $A_f$ for $u\ge 0$. First suppose that $u=0$ (Figure $\ref{diff-5b}(A)$). We have
\begin{equation*}
\begin{split}
S_{i^*+1}(s)-S_{i^*+1}(f) &= \beta \pi_{i^*+1}X-\pi_{i^*+1}X \\
 & \ge6\pi_{i^*+1} \qquad (\because \beta X\ge X+6)
\end{split}
\end{equation*} 

This implies that agent $A_s$ fully executes at least $5$ bits of stage $i^*+1$ during the execution of stage $i^*+2$ of $A_f$. Since $\pi_{i^*+2}\ge 4\pi_{i^*+1}$, agent $A_s$ fully executes at least $3$ consecutive bits of stage $i^*+1$ during the execution of the first bit of stage $i^*+2$ of $A_f$. At least one of those 3 bits is 1. Since during the execution of the first bit of any stage agents remains passive, agent $A_s$ meets agent $A_f$ during the execution of the stage $i^*+2$ of the latter. 
\begin{figure}[h]
   \vspace*{-.194in}
     \centering
    \includegraphics[scale = .42]{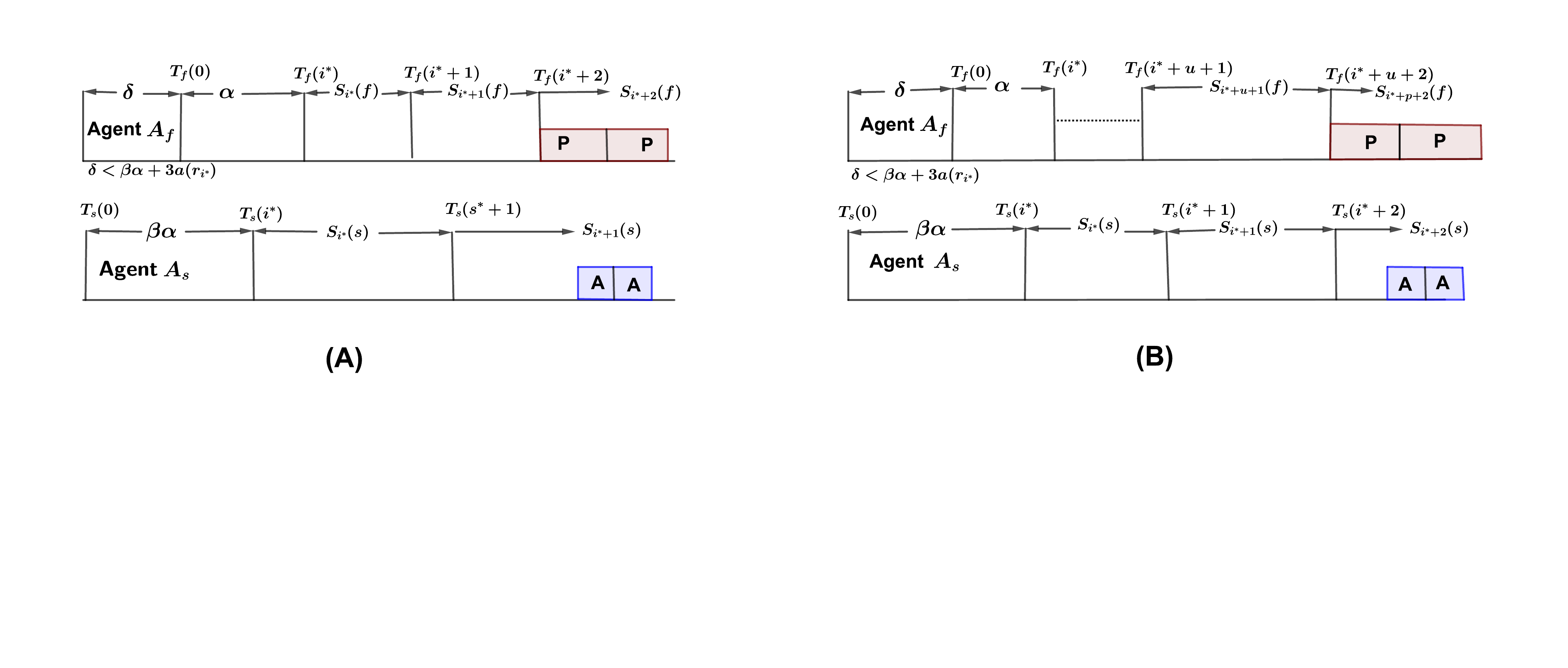}
   \vspace*{-1.2in}
    \caption{An illustration for the proof of Lemma $\ref{lem-11}$ when $p<0$, $T_s(i^*+1)\ge T_f(i^*+1)$ and: (A) agent $A_s$ starts its stage $i^*+1$ during the stage $i^*+1$ of agent $A_f$ i.e., $u=0$, and (B) agent $A_s$ starts its stage $i^*+1$ during the stage $i^*+u+1$ of agent $A_f$ for $u\ge 1$.}
    \label{diff-5b}
\end{figure}

Next consider the case when $u\ge 1$ (Figure $\ref{diff-5b}(B)$). Let $I_s=S_{i^*+2}(s)+ S_{i^*+1}(s)-S_{i^*+1+u}(f)$. Since agent $A_s$ starts its stage $i^*+1$ during the stage $i^*+1+u$ of agent $A_f$ we have

\begin{equation*}
\begin{split}
\beta\alpha+S_{i^*}(s) &\ge\delta+ \alpha+ S_{i^*}(f)+S_{i^*+1}(f)+\cdots+S_{i^*+u}(f) \\
\beta\alpha+S_{i^*}(s) & \ge S_{i^*+u}(f) \qquad \cdots (1)
\end{split}
\end{equation*}

Thus,
\begin{equation*}
\begin{split}
I_s=S_{i^*+2}(s)+S_{i^*+1}(s)-S_{i^*+1+u}(f) &\ge S_{i^*+2}(s)+S_{i^*+1}(s)-5S_{i^*+u}(f) \qquad (\text{by Lemma  \ref{lem-03}})\\
 & \ge S_{i^*+2}(s)+S_{i^*+1}(s)-5\beta\alpha-5S_{i^*}(s) (\text{from  (1)})\\
 & \ge \beta \pi_{i^*+2}X+\beta \pi_{i^*+1}X-\frac{5}{3}\beta\pi_{i^*}X-5\beta\pi_{i^*}X \qquad (\because \alpha<\frac{S_{i^*}(f)}{3})\\
 & \ge \beta X (4^2\pi_{i^*}+4\pi_{i^*}-\frac{20}{3}\pi_{i^*})   \qquad (\text{by Corollary \ref{cor-01}})\\
 & \ge (X+6)13\pi_{i^*} \qquad (\because \beta X\ge X+6)\\
 & \ge 156\pi_{i^*} \qquad (\because X\ge 6)\\
 & \ge 6 \pi_{i^*+2} \qquad (\text{by Corollary \ref{cor-01}})
\end{split}
\end{equation*}  
Hence  by Lemma $\ref{lem-08}$ substituting $q=1$ and $k=u+1$, we can conclude that agent $A_s$ meets agent $A_f$ during the stage at most $i^*+u+2$ of $A_f$ and moreover, during this meeting agent $A_s$ is in stage at most $i^*+2$.

\item $\boldsymbol{T_s(i^*+1)<T_f(i^*+1)):}$ By Lemma $\ref{lem-10}$,  we have $T_s(i^*+1)> T_f(i^*)$. Let $I=\delta+\alpha+S_{i^*}(f)-\beta\alpha-S_{i^*}(s)$. Since agent $A_s$ starts stage $i^*+1$ during the stage $i^*$ of $A_f$, we have $I>0$. Note that $I$ is the duration of the part of stage $i^*+1$ of agent $A_s$ that is executed during the execution of stage $i^*$ of $A_f$. Now one of the last two bits of any transformed label is $1$. Thus,  if $I\ge \frac{3}{2}\pi_{i^*}$, then agent $A_f$ executes at least one complete activity cycle of its stage $i^*$ during the execution of the first two passivity cycles in stage $i^*+1$ of $A_s$ corresponding to the first $0$ of its transformed label (Figure $\ref{diff-5c}(A)$). Thus in this case agent $A_f$ meets agent $A_s$ during the stage $i^*+1$ of the latter. Now suppose $I<\frac{3}{2}\pi_{i^*}$ (Figure $\ref{diff-5c}(B)$). Let $I_{i^*+1}=S_{i^*+1}(s)-S_{i^*+1}(f)$. Then $I_{i^*+1}=\beta X\pi_{i^*+1}-\pi_{i^*+1} X\ge 6\pi_{i^*+1}$. The difference $I_{i^*+1}-I$ is the duration of the part of stage $i^*+1$ of $A_s$ executed during the stage $i^*+2$ of agent $A_f$. Since  $I<\frac{3}{2}\pi_{i^*}<\pi_{i^*+1}$, we have  $I_{i^*+1}-I> 5\pi_{i^*+1}$. This implies that during the execution of the first bit $0$ of the stage $i^*+2$ of agent $A_f$, agent $A_s$ executes at least 3 consecutive bits of stage $i^*+1$ (since $\pi_{i^*+2}\ge 4\pi_{i^*+1}$). At least one of these bits must be 1.
Hence while agent  $A_f$ executes two passivity cycles of total length $\pi_{i^*+2}$, agent $A_s$ executes at least two activity cycles of its stage $i^*+1$ and it meets agent $A_f$ during stage $i^*+1$ of $A_s$. 
  \begin{figure}[h]
   \vspace*{-.194in}
     \centering
    \includegraphics[scale = .42]{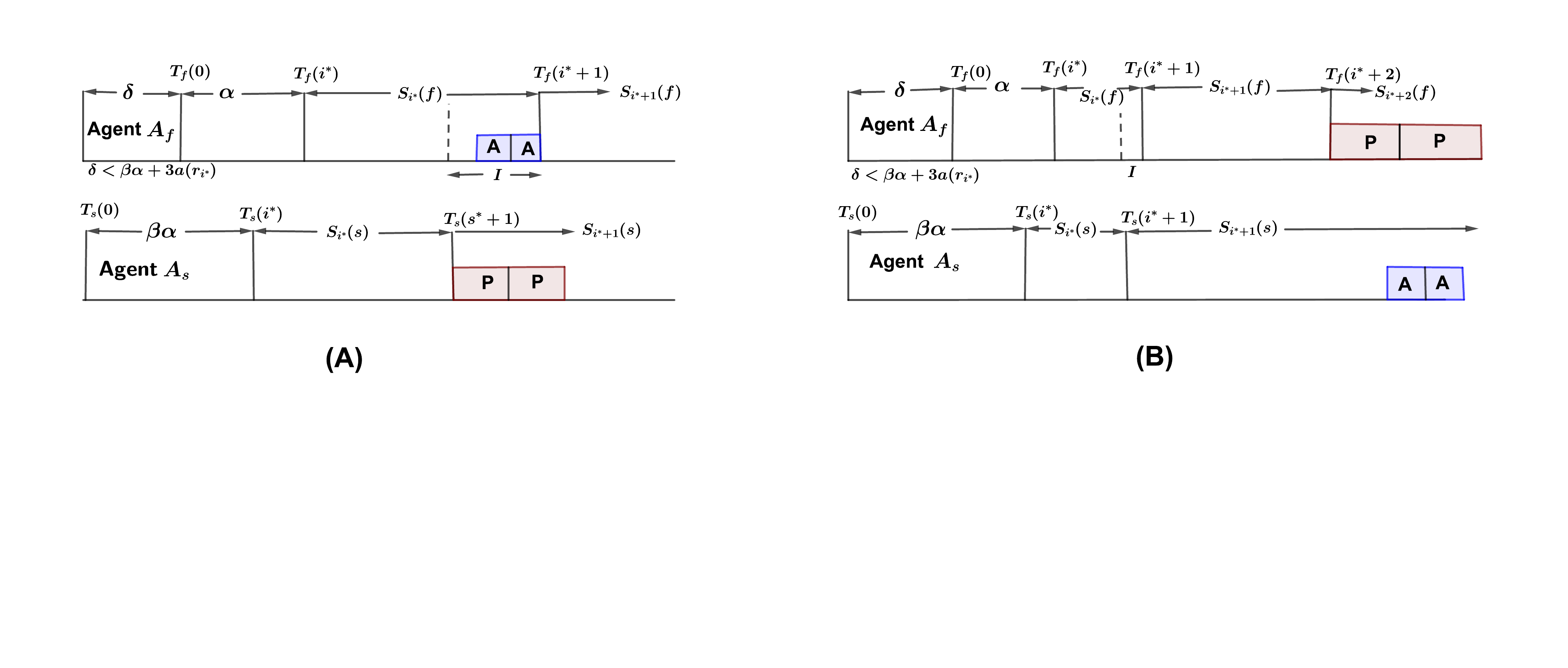}
   \vspace*{-1.2in}
    \caption{An illustration for the proof of Lemma $\ref{lem-11}$ when $p<0$, $T_s(i^*+1)<T_f(i^*+1)$ and: (A) $I\ge \frac{3}{2}\pi_{i^*}$, and (B) $I<\frac{3}{2}\pi_{i^*}$.}
    \label{diff-5c}
\end{figure}

\end{itemize}

\end{itemize}
\end{proof}

\subsubsection{Complexity of Algorithm URT}

The correctness of Algorithm URT follows immediately from Lemmas \ref{lem-07}, \ref{lem-09} and \ref{lem-11}. Indeed, these three lemmas together imply that agents always meet. It remains to estimate the complexity of Algorithm URT. This is done in the following theorem.

\begin{theorem}
Consider two agents with different labels from the set $\{1,\dots,L\}$, executing Algorithm URT in an unoriented infinite regular tree of constant  degree $d\geq 2$, starting at an unknown distance $D$.
Let $z(D)$ be the number of nodes in a ball of radius $D$ in this tree. Then the agents meet in time $O(z(D)\log L)$.
\end{theorem}

\begin{proof}
Recall that the execution time $\tau$ of the algorithm is counted since the start of the earlier agent.
Let $A$ be any agent. 
In view of Lemma \ref{lem-01} and since the length of the transformed label of $A$ is in $O(\log L)$, we have that the duration $S_i(A)$ of stage $i$ of $A$ is in  $O(z(r_i)\log L)$ for any $i$.
Hence each of $S_{i^*}(A)$, $S_{i^*+1}(A)$, $S_{i^*+2}(A)$ are in $O(z(r_i^*)\log L)$, and hence in $O(z(D)\log L)$.
 Let $\alpha(A)$ be the length of time since the start of $A$ to the time when $A$ reaches its critical stage. By Corollary \ref{cor-02}, we have $\alpha(A)\le \frac{S_{i^*}(A)}{3}$, and hence $\alpha(A)$ is also in $O(z(D)\log L)$.
 
 Let $A_1$ be the earlier agent or any of the agents if they start simultaneously. Let $\delta \geq 0$ be the delay of the start of the later agent w.r.t the start of $A_1$. By Lemma \ref{lem-05}, if 
 $\delta  \geq \alpha(A_1)+3a(r_{i^*})$ then agents meet during the critical stage of $A_1$. In this case, $\tau \leq \alpha(A_1)+S_{i^*}(A_1)$ and hence $\tau$ is in $O(z(D)\log L)$.
 
 Hence we may assume that  $\delta < \alpha(A_1)+3a(r_{i^*})$. Thus $\delta$ is in $O(z(D)\log L)$.  If agents have labels of equal length then by Lemma \ref{lem-07} we have $\tau \leq \alpha(A_1)+S_{i^*}(A_1)$ and hence $\tau$ is in $O(z(D)\log L)$.
 Suppose that agents have labels of different lengths and that $A_f$ is the faster agent and $A_s$ is the slower agent. By Lemma  \ref{lem-09}, if $A_f=A_1$ then 
 $\tau\leq \delta + \alpha(A_s)+S_{i^*}(A_s) +S_{i^*+1}(A_s) + S_{i^*+2}(A_s)$. Hence in this case $\tau$ is in $O(z(D)\log L)$. Finally, by Lemma \ref{lem-11}, if $A_s=A_1$ then
 $\tau\leq \alpha(A_s)+S_{i^*}(A_s) +S_{i^*+1}(A_s) + S_{i^*+2}(A_s)$. Hence in this case $\tau$ is in $O(z(D)\log L)$.  This proves that in all cases $\tau$ is in $O(z(D)\log L)$, and hence concludes the proof.
\end{proof}

\subsection{The lower bound}

Algorithm URT has the advantage of working without any extra assumptions: agents may start with arbitrary delay and do not need any knowledge of the parameters of the problem which are the initial distance $D$ between them and the size $L$ of the space of labels. However, the disadvantage of the algorithm is its complexity that has as a factor the number $z(D)$ of nodes in a ball of radius $D$ in the underlying tree. For any degree $d>2$ of the tree, the value of $z(D)$ is exponential in $D$, thus making the time of the algorithm prohibitively large for large $D$. Hence it is natural to ask if such a long time is actually needed for rendezvous. In this section we show that the answer is yes. In fact we prove that  $\Omega(z(D))$ is a lower bound on the time of rendezvous, even in the most benign scenario, when the agents start simultaneously and know the exact values of $D$ and $L$.

\begin{theorem}
For any $d\geq2$ there exists a port labeling of the infinite regular tree of degree $d$, such that the time of any rendezvous algorithm is in $\Omega(z(D))$, even if agents start simultaneously and know the exact values of $D$ and $L$.
\end{theorem}

\begin{proof}
For $d=2$ the theorem is trivial because then $z(D) \in \Theta(D)$. Hence we may assume that $d>2$. Consider the port numbering of the infinite regular tree of degree $d$ in which port numbers at both ends of each edge are equal. For any $d>2$, there is exactly one such tree, up to isomorphism. Notice that an agent executing any rendezvous algorithm in this tree cannot learn anything during the execution. Indeed, when the agent takes port $p$ at some node, it knows in advance that it will enter the adjacent node of the same degree $d$ by a port with the same number $p$. Hence a rendezvous algorithm in this tree does not have any ``if'' statements. It is simply a sequence of terms from $\{0,\dots, d-1\}$ corresponding to ports taken at consecutive steps. This sequence may depend on the label of the agent, but for a given label it is fixed.

Consider any initial distance $D>0$ and any size $L>1$ of the space of labels. Suppose that there exists an algorithm $\cal A$ guaranteeing rendezvous in time at most $t<z(D)/2$ for any agents with different labels starting simultaneously at distance $D$. Consider any labels $\ell_1 \neq \ell_2$ and let $(a_1,\dots,a_t)$ and $(b_1,\dots,b_t)$ be the sequences of integers from $\{0,\dots, d-1\}$ corresponding to the executions of the algorithm by agents $A_1$ and $A_2$ with labels $\ell_1$ and $\ell_2$, respectively.  For any node $v$ and any sequence $(c_1,\dots,c_k)$ of port numbers, let $v(c_1,\dots,c_k)$ denote the node which an agent reaches starting from node $v$ and taking consecutive ports $c_1,\dots,c_k$.

Let $v_1$ and $v_2$ be the starting nodes of agents $A_1$ and $A_2$, respectively, and  let $s\leq t$. If agents meet after time $s$, then they are at the same node after $s$ steps, i.e, 
$v_1 (a_1,\dots,a_s)=v_2(b_1,\dots,b_s)$. This implies that $v_2=v_1(a_1,\dots,a_s,b_s,b_{s-1},\dots,b_1)$. Thus, if agents meet after some time $s\leq t$, then, for a fixed starting node $v_1$ of $A_1$, the number of possible starting nodes $v_2$ of agent $A_2$ is at most $t$. However, the number of nodes at distance exactly $D$ from $v_1$ is larger than half of the size $z(D)$ of the ball of radius $D$ centered at $v_1$. Hence there exists a node $v_2$ of $A_2$ at distance $D$ from $v_1$ such that agents $A_1$ and $A_2$ starting from nodes $v_1$ and $v_2$, respectively, and executing algorithm $\cal A$ do not meet after time at most $t$. This is a contradiction and it proves that the time of algorithm $\cal A$ is in $\Omega(z(D))$.
\end{proof}

Another lower bound $\Omega(D\log L)$ on rendezvous time was proved in \cite{DFKP}. It was proved for agents operating in a ring with port numbers $0,1,0,1,0,1,\dots$ in clockwise order, even if agents start simultaneously and know the exact values of $D$ and $L$.
The same proof is valid for agents in the infinite line with the same port numbering. Hence this lower bound also holds for agents in any infinite regular tree of degree $d\geq 2$ containing such an infinite line. This gives the following corollary.

\begin{corollary}
For any $d\geq2$, any rendezvous algorithm working for all infinite regular trees of degree $d$ with  arbitrary port numberings, must have time in $\Omega(z(D)+D\log L)$, even if agents start simultaneously and know the exact values of $D$ and $L$.
\end{corollary}

Notice that for $d=2$, i.e., for the infinite line, where $z(D)$ is in $ \Theta(D)$, the above corollary implies that Algorithm URT has optimal complexity.

\section{Oriented trees}

An oriented tree is a tree such that one node called the root has label $R$, all other nodes do not have labels, and at each node different from $R$ the port 0 is on the simple path toward $R$. In this section we investigate rendezvous in infinite oriented trees. First note that Algorithm URT designed for rendezvous in unoriented regular trees works for oriented regular trees with the same complexity. However, we will show that in many cases orientation allows us to significantly speed up rendezvous.

We start with the observation that if the delay between the starting times of agents can be arbitrary then there is a large lower bound on rendezvous time even in oriented regular trees. This lower bound holds even if agents know the exact values of the initial distance $D$ and of the size of the label space $L$.

\begin{proposition}
For any $d\geq2$,  the time of any rendezvous algorithm working for agents with arbitrary delay in infinite oriented regular trees of degree $d$ is in $\Omega(z(D))$, even if agents know the exact values of $D$ and $L$.
\end{proposition}

\begin{proof}
%For $d=2$ the theorem is trivial because then $z(D) \in \Theta(D)$. Hence we may assume that $d>2$. 
Consider two agents at distance $D$. The adversary starts one of the agents and delays the start of the other agent by $\delta=z(D)$. The earlier agent cannot visit all nodes at distance $D$ from its starting node by time $\delta$. Let $v$ be any such node not visited by time $\delta$. If the initial node of the later agent is $v$ then the execution time of algorithm $\cal A$ is at least $z(D)$. 
\end{proof}

The lower bound $\Omega(D\log L)$ from \cite{DFKP} holds for infinite oriented regular trees as well (this bound holds even for simultaneous start). Hence we get

\begin{corollary}
For any $d\geq2$, any rendezvous algorithm working for agents with arbitrary delay in infinite oriented regular trees of degree $d$, must have time in $\Omega(z(D)+D\log L)$, even if agents know the exact values of $D$ and $L$.
\end{corollary} 

Thus for agents starting with arbitrary delay, rendezvous must be slow, even for oriented regular trees. It turns out that for simultaneous start, the situation changes dramatically. Recall that for unoriented trees, rendezvous must be slow even for simultaneous start. By contrast, we will show that orientation helps to speed up rendezvous significantly in this case.  Indeed, we design three algorithms. 
All of them work for arbitrary oriented trees (regularity does not have to be assumed).
Two of them work in the optimal time $O(D\log L)$ under assumption that agents know some upper bounds on  the parameters $D$ or $L$. The third algorithm works without any extra assumptions and has the complexity $O(D^2+\log^2 L)$, hence it still avoids the extremely costly lower bound $\Omega(z(D))$.  

The main idea of these three rendezvous algorithms is the same. We have to avoid costly exploration of balls that was crucial in Algorithm URT. Here is where we use orientation. 
Agents go ``up'' (towards the root) in order to position themselves on the same branch and then make moves up and down on this branch, depending on bits of their transformed label. These moves play a role similar to ball exploration for unoriented trees but are much faster, and permit the agents to meet when the range of the moves is sufficiently large.

In all these algorithms we will use the following procedures that take a positive integer $x$ as parameter. Procedure $Up(x)$ consists of $x$ consecutive moves, each of them taking port 0, unless the root $R$ is visited on the way. In the latter case the agent stops at $R$ and stays there forever. Thus the procedure results in going $x$ steps towards the root or getting to the root. Procedure $Up-and-Down(x)$ consists of an execution of Procedure $Up(x)$ followed by a backtrack to the node where Procedure $Up(x)$ started. Procedure $Up(x)$ lasts $x$ rounds and Procedure $Up-and-Down(x)$ lasts $2x$ rounds.

\subsection{Known polynomial bound $L^*$ on label space size $L$}

In this section we assume that the agents know some common polynomial upper bound $L^*$ on the size $L$ of the label space.
Let $\lambda = \lceil \log L^* \rceil +1$. Hence the length of the binary representation of all labels is at most $\lambda$ and the agents know $\lambda$. 
Moreover, since $L^*$ is a polynomial upper bound on $L$, we have that $\lambda$ is in $O(\log L)$.
For any label $\ell \in \{1,\dots, L\}$ we define the {\em padded label} $Pad(\ell)$ as follows. Let $(c_1,\dots, c_s)$ be the binary representation of $\ell$. $Pad(\ell)$ is the binary sequence of length $2\lambda$ obtained as follows. First concatenate a prefix of $\lambda -s$ zeroes to $(c_1,\dots, c_s)$ to get a sequence of length $\lambda$, and then replace each bit 1 by 10 and each bit 0 by 01.   All padded labels have equal length $2\lambda$ and the following property. If $\ell_1 \neq \ell_2$ then there exists an index $j\leq 2\lambda$, such that the $j$th bit of 
$Pad(\ell_1)$ is 1 and the $j$th bit of  $Pad(\ell_2)$ is 0. 

Algorithm Known-Bound-on-L works in stages $i=0,1,\dots$. In a stage $i$ the agent makes $2^i$ steps up and then ``executes'' consecutive bits of its padded label: if the bit is 1, the agent executes Procedure $Up-and-Down(2^i)$, and it stays idle for $2\cdot 2^i$ rounds if the bit is 0. There is an exception to this rule: if an agent visits the root $R$, it stops executing moves and stays there forever.
In view of the simultaneous start, the use of padded labels whose length is the same for both agents guarantees that both agents will start each stage simultaneously (unless one of them visits $R$). 
The algorithm is interrupted as soon as the agents meet. 

{\bf Algorithm Known-Bound-on-L}

\hspace*{1cm}$Pad(\ell):= (b_1,b_2,\dots,b_k)$\\
\hspace*{1cm}{\bf for} $i=0,1,2,\dots$ {\bf do}\\
\hspace*{2cm}$Up(2^i)$\\
\hspace*{2cm}{\bf for} $j=1$ {\bf to} $k$ {\bf do}\\
\hspace*{3cm}{\bf if} $b_j=1$ {\bf then}\\
\hspace*{4cm}$Up-and-Down(2^i)$\\
\hspace*{3cm}{\bf else}\\
\hspace*{4cm}stay idle for $2\cdot 2^i$ rounds

\begin{theorem}
Consider two agents with different labels from the set $\{1,\dots,L\}$, executing Algorithm Known-Bound-on-L in an infinite oriented tree, starting simultaneously at an unknown distance $D$. Suppose that agents know a common polynomial  upper bound $L^*$ on the size $L$ of the label space. 
Then the agents meet in time $O(D\log L)$, which is optimal.
\end{theorem}

\begin{proof}
Let $i^*$ be the smallest integer $i$ such that $D\leq 2^i$.
First suppose that none of the agents visits the root $R$ before the end of stage $i^*$.
The duration of stage $i$ is $(4\lambda +1)2^i$, for each of the agents. Since agents start simultaneously, they also start each stage $i$ simultaneously. By induction on $i$, after the first execution of Procedure $Up(2^i)$ of stage $i$, agents are at distance at most $D$ and each agent ends stage $i$ at the same node at which it was after the first execution of Procedure $Up(2^i)$ in this stage.     Consider the time $t$ immediately after the first execution of Procedure $Up(2^{i^*})$ of stage $i^*$. Both agents are on the same branch in the tree, at distance at most $D$. Let $A_1$ be the lower agent and $A_2$ the upper agent (i.e., $A_2$ is on the simple path from the position of $A_1$ at time $t$ to the root). Let $j\leq 2\lambda$ be the first index such that the $j$th bit of the padded label of $A_1$ is 1 and the $j$th bit of the padded label of $A_2$ is 0. Both agents execute this $j$th bit simultaneously during $2\cdot 2^{i^*}$ rounds. Hence the upper agent $A_2$ stays idle at distance at most $D$ while the lower agent $A_1$ executes the first part of Procedure $Up-and-Down(2^{i^*})$. In view of $D\leq 2^{i^*}$ this guarantees that agent $A_1$ meets agent $A_2$. This meeting occurs during stage $i^*$, i.e., in time $O(2^{i^*}\lambda)$ which is $O(D\log L)$, by definition of $\lambda$ and $i^*$.

Next suppose that one of the agents visits root $R$ before the end of stage $i^*$. Suppose that agent $A_1$ is the first to visit this node. Let $g_1$ and $g_2$ be  the distance of the starting node of agent $A_1$ (resp. $A_2$) from the root. Hence $g_2\leq g_1 +D$. Hence, if agents do not meet before, agent $A_2$ must reach the root $R$ by the end of stage $i^*+1$ and stop. Hence the meeting must occur at the latest in time $O(2^{i^*}\lambda)$ which is $O(D\log L)$, by definition of $\lambda$ and $i^*$.
\end{proof}

\subsection{Known linear bound $D^*$ on the initial distance $D$}

In this section we assume that the agents know some common linear upper bound $D^*$ on the initial distance $D$ between them but they may have no knowledge about $L$.
For any label $\ell \in \{1,\dots, L\}$ we first define the {\em prefix-free label}  $PF(\ell)$ as follows. Let $(c_1,\dots, c_s)$ be the binary representation of $\ell$. In order to obtain $PF(\ell)$,
replace each bit 1 by 10, each bit 0 by 01 and add bits 11 at the end. The obtained sequence is of length $2s+2$ and has the property that if we start with two different labels then none of the obtained sequences can be a prefix of the other  (cf. \cite{DFKP}). To get the {\em adapted label} $Adapt(\ell)$, replace each 1 by 10 and each 0 by 01 in the string $PF(\ell)$.
The resulting sequence $Adapt(\ell)$ has length $4s+4$. Notice that since binary representations of labels may have different lengths, the same is true for the adapted labels. However, the adapted labels have the property that if $\ell_1 \neq \ell_2$ then there exists an index $j$, such that the $j$th bit of 
$Adapt(\ell_1)$ is 1 and the $j$th bit of  $Adapt(\ell_2)$ is 0. 

Algorithm Known-Bound-on-D has the same idea as Algorithm Known-Bound-on-L but now we do not need stages, as a linear bound $D^*$ on $D$ is known: agents  know the appropriate search range from the outset.

{\bf Algorithm Known-Bound-on-D}

\hspace*{1cm}$Adapt(\ell):= (b_1,b_2,\dots,b_k)$\\
\hspace*{1cm}$Up(D^*)$\\
\hspace*{1cm}{\bf for} $j=1$ {\bf to} $k$ {\bf do}\\
\hspace*{2cm}{\bf if} $b_j=1$ {\bf then}\\
\hspace*{3cm}$Up-and-Down(D^*)$\\
\hspace*{2cm}{\bf else}\\
\hspace*{3cm}stay idle for $2 D^*$ rounds\\
\hspace*{1cm}{\bf if} the current node is not $R$ {\bf then}\\
\hspace*{2cm}$Up(D^*)$\\

\begin{theorem}
Consider two agents with different labels from the set $\{1,\dots,L\}$, executing Algorithm Known-Bound-on-D in an infinite oriented tree, starting simultaneously at a  distance $D$. Suppose that agents know a common linear  upper bound $D^*$ on $D$. 
Then the agents meet in time $O(D\log L)$.
\end{theorem}

\begin{proof}
First suppose that none of the agents visits the root $R$ by the end of the first execution of procedure $Up(D^*)$.
Consider the time $t$ immediately after the first execution of Procedure $Up(D^*)$. Both agents are on the same branch in the tree, at distance at most $D$. Let $A_1$ be the lower agent and $A_2$ the upper agent. Let $j$ be the first index such that the $j$th bit of the adapted label of $A_1$ is 1 and the $j$th bit of the adapted label of $A_2$ is 0. Both agents execute this $j$th bit simultaneously during $2D^*$ rounds. Hence the upper agent $A_2$ stays idle at distance at most $D$ while the lower agent $A_1$ executes the first part of Procedure $Up-and-Down(D^*)$. In vew of $D\leq D^*$, this guarantees that agent $A_1$ meets agent $A_2$. This meeting occurs in time $O(D^*k)$ which is $O(D\log L)$, by definition of $k$ and $D^*$.

Next suppose that one of the agents visits the root $R$ by the end of the first execution of procedure $Up(D^*)$. Suppose that agent $A_1$ is the first to visit this node. Let $g_1$ and $g_2$ be  the distance of the starting node of agent $A_1$ (resp. $A_2$) from the root. Hence $g_2\leq g_1 +D$. It follows that if agent $A_2$ did not visit root $R$ by the end of the first execution of procedure $Up(D^*)$ then it must visit it by the end of the second execution of procedure $Up(D^*)$. Hence the meeting occurs in time $O(D^*k)$ which is $O(D\log L)$, by definition of $k$ and $D^*$.
\end{proof}

\subsection{No extra knowledge}

We finally consider the situation when agents start simultaneously but do not have any extra knowledge about the parameters $D$ and $L$. In this case we design a rendezvous algorithm which, although slower than the two previous ones, still avoids the exponential lower bound $z(D)$ that was unavoidable for unoriented trees even with simultaneous start and also unavoidable for oriented trees with arbitrary delay.

First, for any label $\ell$, define an infinite binary sequence $Adapt^*(\ell)$. The sequence is defined as follows. Concatenate infinitely many copies of the prefix-free label $PF(\ell)$ (defined in the previous subsection) and denote the obtained sequence by $PF^*(\ell)$. Then $Adapt^*(\ell)$ is defined by replacing each 1 by 10 and each 0 by 01 in $PF^*(\ell)$. Note that an equivalent way of defining $Adapt^*(\ell)$ is to concatenate infinitely many copies of the adapted label $Adapt(\ell)$, defined in the previous subsection.

The idea of Algorithm No-Extra-Knowledge is similar  to that of Algorithms Known-Bound-on-L and Known-Bound-on-D but with two important differences. First, in each stage only one bit of
$Adapt^*(\ell)$ is processed and second, in consecutive stages the ranges of agents' moves are not doubled but increased by 1. The algorithm is interrupted as soon as the agents meet. 

{\bf Algorithm No-Extra-Knowledge}

\hspace*{1cm}$Adapt^*(\ell):= (b_1,b_2,\dots)$\\
\hspace*{1cm}{\bf for} $j=1,2,\dots$ {\bf do}\\
\hspace*{2cm}$Up(j)$\\
\hspace*{2cm}{\bf if} $b_j=1$ {\bf then}\\
\hspace*{3cm}$Up-and-Down(j)$\\
\hspace*{2cm}{\bf else}\\
\hspace*{3cm}stay idle for $2j$ rounds\\

In the proof of correctness and the analysis of complexity of Algorithm No-Extra-Knowledge we will use the following lemma.

\begin{lemma}\label{adapt}
Consider two distinct labels $\ell_1$ and $\ell_2$. For any index $j\geq 4$ there exists an integer $y\leq 4\log L$, such that the $(j+y)$th bit of $Adapt^*(\ell_1)$ is 1 and the $(j+y)$th bit of $Adapt^*(\ell_2)$ is 0.

\end{lemma}

\begin{proof}
First consider the infinite sequences $PF^*(\ell_1)$ and $PF^*(\ell_2)$. Let $x$ be the length of $PF(\ell_1)$.  Let $i\geq 1$ be any even index. Let $i'\geq i$ be the smallest index at which a copy of $PF(\ell_1)$ starts in $PF^*(\ell_1)$. By definition of a prefix-free label, there exists an index $i''\leq i'+x$ such that the $i''$th bit is different in $PF^*(\ell_1)$ and $PF^*(\ell_2)$. Let $p_1q_1$ and
$p_2q_2$ be the pairs of bits in $Adapt^*(\ell_1)$ and $Adapt^*(\ell_2)$ respectively, corresponding to the $i''$th bit of $PF^*(\ell_1)$ and $PF^*(\ell_2)$ respectively. Hence, either $(p_1q_1)=(10)$
and $(p_2q_2)=(01)$ or vice-versa. In both cases, for one of the indices $t$ corresponding to these bits, the $t$-th bit of  $Adapt^*(\ell_1)$ is 1 and the $t$-th bit of  $Adapt^*(\ell_2)$ is 0.

Consider any index $j\geq 4$. Let $i=2\lfloor j/4\rfloor$. Hence we have
$i''\leq i'+x \leq i+2\log L$ and thus $t\leq 2i''\leq 2i+4\log L\leq j+4\log L$,
which concludes the proof.
\end{proof}

We are now ready to prove the correctness and analyze the complexity of Algorithm No-Extra-Knowledge.

\begin{theorem}
Consider two agents with different labels from the set $\{1,\dots,L\}$, executing Algorithm No-Extra-Knowledge in an infinite oriented tree, starting simultaneously at a  distance $D$. 
Then the agents meet in time $O(D^2+\log^2 L)$.
\end{theorem}

\begin{proof}
First assume that none of the agents visits the root $R$ before the meeting.
Since agents start simultaneously, they execute each turn of the {\bf for} loop (corresponding to the consecutive bits) precisely in the same time period. It is enough to prove the theorem for $D\geq 4$. Consider any index $j\geq D$. After the execution of the $j$-th bit
both agents are on the same branch in the tree, at the same distance at most $D$. 
Let $\ell_1$ be the label of the lower agent.
By Lemma \ref{adapt}, there exists an integer $y\leq 4\log L$, such that the $(D+y)$th bit of $Adapt^*(\ell_1)$ is 1 and the $(D+y)$th bit of $Adapt^*(\ell_2)$ is 0. During the execution of bit $t=D+y$, first both agents execute procedure $Up(t)$ simultaneously, and then the upper agent stays idle for $2t$ rounds, while the lower agent executes procedure $Up-and-Down(t)$. Since $t\geq D$, the lower agent must meet the upper agent during the first half of the execution of this procedure. It follows that the agents meet at the latest after the execution of bit $D+\lceil 4\log L\rceil$. Since the time of execution of any bit $j$
is $3j$, executing all bits until bit $D+\lceil 4\log L\rceil$ takes time $O((D+\lceil 4\log L\rceil)^2)=O(D^2+\log^2L)$.

Next assume that one of the agents visits the root $R$ before the meeting. Consider the first agent visiting $R$. This must happen at some time $s$ before the end of the execution of bit $D+\lceil 4\log L\rceil$. The other agent must reach $R$ by the time $s+D$. Since $s$ is in $O(D^2+\log^2L)$, the meeting at $R$ must also happen in time $O(D^2+\log^2L)$.
\end{proof}

{\bf Remark.}
It is important to explain why the ranges of agents processing consecutive bits are incremented by 1 and not doubled, as in the other algorithms. If the range in the processing of the $j$-th bit was $2^j$ instead of $j$, then processing bit $D+\lceil 4\log L\rceil$ would take time $O(2^{D+\lceil 4\log L\rceil})$ causing an exponential blow-up. Incrementing by 1 results in total cost quadratic in $D+\log L$ but not exponential.
We have seen before that when some upper bound on $D$ or on $L$ is known, this problem can be avoided altogether, resulting in optimal complexity $O(D\log L)$.

\section{Conclusion}

We studied deterministic rendezvous in unoriented and oriented infinite trees, showing the impact of orientation of the tree on the time of rendezvous. For unoriented regular trees, we designed an algorithm working in time $O(z(D)\log L)$ and showed a lower bound of $\Omega(z(D)+D\log L)$ on the time of any such algorithm. While these bounds match for $d=2$ and are very close for $d>2$ if $L$ is not very large, there remains a gap whose filling is a natural open problem. For oriented regular trees with arbitrary delay between waking times of agents, the situation is identical: the above algorithm still works and the same lower bound is valid. However, for simultaneous start in oriented trees (not necessarily regular), the situation is different. Assuming either the knowledge of a polynomial upper bound on $L$ or of a linear upper bound on $D$, we showed algorithms working in time $O(D\log L)$, which is optimal.  Without such extra knowledge, we showed an algorithm working in time $O(D^2+\log^2L)$, thus avoiding the exponential lower bound $\Omega(z(D))$ without additional assumptions. If $D$ and $\log L$ are of the same order of magnitude, this algorithm is still optimal.
The problem whether there exists a rendezvous algorithm working in oriented trees in time $O(D\log L)$ for simultaneous start,
 without assuming any knowledge concerning $D$ and $L$, remains open.

It remains to discuss our assumptions concerning the environment where the agents operate. First it is natural to ask if our results could be generalized for graphs which are not trees. Our main rendezvous algorithm for unoriented regular trees relies on exploration of balls of increasing radii, with the aim of meeting the other agent that stays idle at some node of such a ball during its exploration. It should be recalled that while in an anonymous tree a ball can be explored in time linear in its size, no such exploration algorithm is known in arbitrary anonymous graphs because of loops that the agent may not notice due to anonymity of the nodes. Hence most rendezvous algorithms known from the literature and working for anonymous graphs (cf. \cite{DFKP,KM,TSZ07}) are restricted to finite graphs, rely on the exploration of the entire graph, and have complexity polynomial in its size. Efficient algorithms working for infinite graphs \cite{BCGIL,CCGKM} use additional assumptions, such as the knowledge of the position of the agent in the graph. It remains open to design efficient rendezvous algorithms working in arbitrary infinite graphs without such extra knowledge.

Restricting attention to trees, we should discuss the assumptions that they are infinite and, in case of unoriented trees, regular. The first assumption is for convenience of problem statement, in order to dismiss algorithms relying on the exploration of the entire tree. In \cite{DFKP}, the authors give such an algorithm working in time $O(n+\log L)$, where $n$ is the size of the tree. 
%Of course, our algorithms also work without any change and with the same complexity for ``large'' finite regular trees (trees with all internal nodes of the same degree) where the agents start at nodes sufficiently far from the leaves. 
In our case, the assumption that the tree is infinite can be easily removed. All our algorithms work for finite trees as well, without change of complexity. (In the case of finite trees, ``regular'' means that all {\em internal} nodes have the same degree.)
%Our algorithm for unoriented regular trees can be easily modified to work for arbitrary finite regular trees (trees with all internal nodes of the same degree), without change of complexity. Our algorithms for oriented trees work for finite oriented trees without any change in the algorithms and in the complexity.

The discussion of the second assumption, that of regularity of the tree, is more subtle. In our algorithm for unoriented trees, we use it in order to guarantee that all balls of a given radius have equal size, and thus an agent can compute the periods of its idleness to match exploration periods of the other agent. This is not possible if the tree is arbitrary because then different balls of the same radius can have very different sizes. If agents knew a common upper bound $b(s)$ on the size of any ball of radius $s$ then our algorithm for unoriented trees could be generalized to arbitrary (non-regular) trees with $z(D)$ replaced by $b(D)$ in the complexity. 
%One way to circumvent the problem of the lack of knowledge of $b(s)$ would be to first explore, in stage $i$, the ball $B(v,2r_i)$ of radius $2r_i$ centered at the initial node $v$ of the agent, thus learning the maximum size of a ball $B(w,r_i)$ of radius $r_i$, centered at any node $w$ in $B(v,r_i)$, and use it to determine the waiting time in stage $i$. However, this would be very inefficient for trees of large degrees, as the size of a ball of radius $2r$ is then much larger than the size of a ball of radius $r$.
However, for arbitrary unoriented trees without any extra knowledge, designing a rendezvous algorithm with time close to optimal seems to be a challenging open problem. By contrast, all our algorithms with simultaneous start for oriented trees are efficient and work without assuming regularity, as these algorithms do not rely on exploring balls in the tree.

%%%%%%%%%%%%%%%%%%%%%%%%%%%%%%%%%%%%%%%%%%%%%%%%%%%%%%%%%%%
\bibliographystyle{plain}

\end{document}